\documentclass[sigconf]{acmart}

\usepackage{booktabs} 
\usepackage{IEEEtrantools}
\usepackage{natbib}
\usepackage{latexsym,amsfonts,amsmath}
\usepackage{mathrsfs}
\usepackage{graphicx}
\usepackage{wrapfig}
\usepackage{paralist}
\usepackage{dsfont}
\usepackage{eso-pic}
\usepackage{epsfig}
\usepackage{ragged2e}
\usepackage{mathtools}
\usepackage{epstopdf}
\usepackage{lipsum}
\usepackage{optidef}
\usepackage{varwidth}
\usepackage{fancyhdr}
\usepackage{subfig}

\newtheorem{problem}{Problem}

\newtheorem{remark}{Remark}
\newtheorem{deff}{Definition}
\newtheorem{thm}{Theorem}

\newtheorem{coro}{Corollary}
\newtheorem{prop}{Proposition}
\newtheorem{lem}{Lemma}

\usepackage{tikz}
\usetikzlibrary{calc,shapes,arrows,arrows.meta, positioning}

\usepackage{algorithm}
\usepackage{algorithmic}

\usepackage{xcolor}         
\definecolor{rangyek}{RGB}{0, 75, 255}
\definecolor{rangdo}{RGB}{239, 62, 91}

\definecolor{START}{rgb}{0.4660 0.6740 0.1880}
\definecolor{TARGET}{rgb}{0 0.4470 0.7410}
\definecolor{OBSTACLES}{rgb}{0.93,0.69,0.13}

\definecolor{START1}{rgb}{0, 1, 0}
\definecolor{TARGET1}{rgb}{0.3010 0.7450 0.9330}
\definecolor{OBSTACLES1}{rgb}{1, 0, 0}

\usepackage[utf8]{inputenc} 
\usepackage[T1]{fontenc}    
\usepackage{hyperref}       
\usepackage{xpatch}
\usepackage{etoolbox}
\hypersetup{
	colorlinks = true,
	linkcolor = rangyek,
	citecolor = rangdo
}
\makeatletter
\newcommand{\reff}[1]{\textcolor{purple}{\ref{#1}}}
\renewcommand{\ref}[1]{\hyperref[#1]{\reff{#1}}}
\makeatother

\renewcommand{\eqref}[1]{\textnormal{(\hyperref[#1]{\textcolor{purple}{\ref*{#1}}})}}

\let\oldcitealp\citealp
\renewcommand{\cite}[1]{[\textcolor{TARGET}{\oldcitealp{#1}}]}

\xpretocmd{\maketitle}{\hypersetup{
		colorlinks = true,
		linkcolor = rangyek,
		citecolor = rangdo,
		urlcolor = rangyek,
		filecolor = rangdo,
		allcolors = rangyek
}}{}{}
\usepackage{url}            
\usepackage{booktabs}       
\usepackage{amsfonts}       
\usepackage{nicefrac}       
\usepackage{microtype}      

\usepackage{stackengine}
\usepackage{scalerel}

\usepackage[many]{tcolorbox}
\usetikzlibrary{calc}
\tcbuselibrary{skins}

\makeatletter
\newsavebox\myboxA
\newsavebox\myboxB
\newlength\mylenA

\newcommand*\xoverline[2][0.75]{%
	\sbox{\myboxA}{$\m@th#2$}%
	\setbox\myboxB\null
	\ht\myboxB=\ht\myboxA%
	\dp\myboxB=\dp\myboxA%
	\wd\myboxB=#1\wd\myboxA
	\sbox\myboxB{$\m@th\overline{\copy\myboxB}$}
	\setlength\mylenA{\the\wd\myboxA}
	\addtolength\mylenA{-\the\wd\myboxB}%
	\ifdim\wd\myboxB<\wd\myboxA%
	\rlap{\hskip 0.5\mylenA\usebox\myboxB}{\usebox\myboxA}%
	\else
	\hskip -0.5\mylenA\rlap{\usebox\myboxA}{\hskip 0.5\mylenA\usebox\myboxB}%
	\fi}
\makeatother

\newtcolorbox{resp}[1][]{%
	enhanced jigsaw,%
	colback=gray!5!white,%
	colframe=gray!80!black,%
	size=small,%
	boxrule=1pt,%
	halign title=flush center,%
	coltitle=black,%
	breakable,%
	drop shadow=black!50!white,%
	attach boxed title to top left={xshift=1cm,yshift=-\tcboxedtitleheight/2,yshifttext=-\tcboxedtitleheight/2},%
	minipage boxed title=3cm,%
	boxed title style={%
		colback=white,%
		size=fbox,%
		boxrule=1pt,%
		boxsep=2pt,%
		underlay={%
			\coordinate (dotA) at ($(interior.west) + (-0.5pt,0)$);
			\coordinate (dotB) at ($(interior.east) + (0.5pt,0)$);
			\begin{scope}[gray!80!black]
				\fill (dotA) circle (2pt);
				\fill (dotB) circle (2pt);
			\end{scope}
		}%
	},%
	#1%
}

\newtcolorbox{mybox}[2][]{enhanced,
	attach boxed title to top left={xshift=1cm,yshift=-2mm},
	fonttitle=\bfseries,varwidth boxed title=0.7\linewidth,
	colbacktitle=gray!45!white,coltitle=gray!10!black,colframe=gray!50!black,
	interior style={top color=gray!5!white,bottom color=gray!0!white},
	boxed title style={boxrule=0.75mm,colframe=white,
		borderline={0.1mm}{0mm}{gray!50!black},
		borderline={0.1mm}{0.75mm}{gray!50!black},
		interior style={top color=gray!30!white,bottom color=gray!5!white,
			middle color=gray!50!white},
		drop fuzzy shadow},
	title={#2},#1}

\usepackage{xspace}
\newcommand{\R}{{\mathds{R}}}
\newcommand{\Rp}{{\mathds{R}_{> 0}}}
\newcommand{\Rpz}{{\mathds{R}_{\geq 0}}}
\newcommand{\N}{{\mathds{N}}}
\newcommand{\Np}{{\mathds{N}_{\geq 1}}}
\newcommand{\I}{{\mathds{I}}}
\newcommand{\Xo}{{\mathcal{X}_{\mathrm{o}}}}
\newcommand{\Xu}{{\mathcal{X}_{\mathrm{1}}}}

\newcommand{\redsquare}{\tikz\fill[red!25] (0,0) rectangle (2mm,2mm);}

\definecolor{DD}{rgb}{0.30, 0.34, 0.93}
\definecolor{DDD}{rgb}{0, 1, 0.76}

\newcommand{\bluesquare}{\tikz\fill[DD!25] (0,0) rectangle (2mm,2mm);}
\newcommand{\greensquare}{\tikz\fill[DDD!25] (0,0) rectangle (2mm,2mm);}

\begin{document}
	
	\title[Data-Driven Controller Synthesis for Discrete-Time General Nonlinear Systems]{Data-Driven Dynamic Controller Synthesis for Discrete-Time General Nonlinear Systems}
	
	\author{Behrad Samari}
	\affiliation{%
			\institution{School of Computing}
			\institution{Newcastle University}
            \country{}
		}
	\email{behrad.samari@newcastle.ac.uk}
	\author{Abolfazl Lavaei}
	\affiliation{%
			\institution{School of Computing}
		   \institution{Newcastle University}
           \country{}
		}
	\email{abolfazl.lavaei@newcastle.ac.uk}

	\begin{abstract}
		Synthesizing safety controllers for general nonlinear systems is a highly challenging task, particularly when the system models are unknown, and input constraints are present. While some recent efforts have explored data-driven safety controller design for nonlinear systems, these approaches are primarily limited to specific classes of nonlinear dynamics (\emph{e.g.,} polynomials) and are not applicable to general nonlinear systems. This paper develops a direct data-driven approach for discrete-time \emph{general nonlinear systems}, facilitating the simultaneous learning of control barrier certificates (CBCs) and \emph{dynamic controllers} to ensure safety properties under input constraints. Specifically, by leveraging the \emph{adding-one-integrator approach}, we incorporate the controller’s dynamics into the system dynamics to synthesize a virtual static-feedback controller for the augmented system, resulting in a dynamic safety controller for the actual dynamics. We collect input-state data from the augmented system during a finite-time experiment, referred to as a \emph{single trajectory}. Using this data, we learn augmented CBCs and the corresponding virtual safety controllers, ensuring the safety of the actual system and adherence to input constraints over a finite time horizon. We demonstrate that our proposed conditions boil down to some data-dependent linear matrix inequalities (LMIs), which are easy to satisfy. We showcase the effectiveness of our data-driven approach through two case studies: one exhibiting significant nonlinearity and the other featuring high dimensionality.
	\end{abstract}
	
	
	\maketitle
	
	\section{Introduction}\label{Section: Introduction}
	Controlling real-world dynamic systems, particularly those with safety-critical applications, demands an in-depth grasp of their inherent nonlinear dynamics. These systems often exhibit complex interdependent behaviors that add layers of difficulty to both analysis and controller synthesis. This complexity is particularly evident in practical systems where obtaining an exact mathematical model is rarely possible. In such cases, the absence of a precise model makes the synthesis of reliable controllers challenging, as it requires capturing system behavior accurately across diverse operating conditions.
	
	In response, the control community has increasingly turned to data-driven approaches to address the challenges posed by unknown (or partially known) models. These methodologies utilize observed data to inform both the analysis and design of controllers, offering a practical pathway to address systems with unknown dynamics. In particular, by relying on empirical data, these techniques aim to capture the essential behaviors of the system without requiring a fully specified mathematical model. However, despite these promising advances, ensuring formal safety guarantees for \emph{general nonlinear systems} remains a substantial obstacle. Specifically, the absence of precise models often restricts the ability to provide rigorous assurances, especially in systems with highly complex nonlinearities.
	
	\textbf{State-of-the-art.} The existing literature offers the concept of control barrier certificates (CBCs) to provide safety guarantees for complex dynamic systems with continuous state spaces (see \emph{e.g.,} \cite{prajna2004safety, prajna2007framework, wieland2007constructive,ames2019control,xiao2023safe}). In a nutshell, a CBC is analogous to a Lyapunov function defined over the state space of a dynamical system, establishing a framework of inequality constraints that apply to both its value and its evolution over time, governed by the system's dynamics. Therefore, if an appropriate level set of the CBC can delineate an unsafe region from all possible system trajectories originating from a specified initial set, the existence of such a CBC provides a formal (probabilistic) safety guarantee for the system (see \emph{e.g.,} \cite{ames2016control, clark2021control, santoyo2021barrier, clark2021verification, samari2024multiplicative, zaker2024compositional,lavaei2024scalable,wooding2024protect,nejati2024context}). However, these works fundamentally rely on the assumption that an accurate mathematical model of the system dynamics is available. Yet, in real-world applications, such models are often either unavailable or too complex to be practically useful.
	
	To address this critical issue while also embracing the shift toward data-driven approaches, the existing literature offers two distinct yet valuable data-driven lenses: (i) \emph{indirect} data-driven approaches, in which the model is first identified, followed by designing the controller through powerful model-based techniques (see \emph{e.g.}, \cite{jagtap2020control, dhiman2021control}), and (ii) \emph{direct} data-driven methodologies, which bypass the model identification step and directly learn the desired controller from data \cite{hou2013model, dorfler2022bridging}. In indirect techniques, although sophisticated methods can be applied if the model identification phase is successful, current system identification methodologies have significant limitations when applied to general nonlinear systems \cite{kerschen2006past}, often making the identification process highly complex and time-intensive, if not impossible. Even when an accurate mathematical model of the system is obtained, indirect approaches inherently involve a \emph{two-layer} complexity: first, deriving the model, and second, utilizing model-based techniques to address the control problem.  In contrast, \emph{direct} data-driven methods reduce complexity to a single layer by learning the controller directly from collected data.
	
	Recently and followed by \cite{de2019formulas}, numerous direct data-driven methodologies have been developed to tackle various challenges in designing controllers for dynamical systems with unknown models, including (approximate) nonlinearity cancellation and robust controller design (\emph{e.g.,} \cite{de2023learning, berberich2022combining}). However, a notable limitation of these approaches is their inability to explicitly handle state constraints. Additionally, data-driven methods have emerged, enabling the synthesis of state-feedback controllers that guarantee the \emph{robust invariance} of a compact polyhedral set containing the origin (\emph{e.g.,} \cite{bisoffi2022controller}). While effective, these methods can be overly conservative, as feasible controllers may exist for \emph{certain subsets} of the polyhedral set even when a controller cannot be designed for the entire set. Furthermore, some studies employ \emph{scenario-based approaches} \cite{calafiore2006scenario} to design barrier certificates for control applications (\emph{e.g.,} \cite{nejati2023formal,aminzadeh2024physics}). However, these approaches rely on the assumption that the data are independent and identically distributed (i.i.d.), necessitating the collection of multiple trajectories, potentially up to millions in practical applications. In contrast, our methodology requires only a single input-state trajectory for conducting the control analysis.
	
	More recently, some progress has been made in designing CBCs and their safety controllers \cite{nejati2022data, samari2024single}. While \cite{nejati2022data} and \cite{samari2024single} address continuous- and discrete-time input-affine nonlinear systems with polynomial dynamics, respectively, they are not applicable to systems with \emph{general nonlinear dynamics}, which is the main focus of our work. Furthermore, these methods rely on data-driven sum-of-squares (SOS) optimization, which introduces scalability challenges compared to our proposed LMI conditions that can be efficiently solved using semidefinite programs (SDP). Additionally, while \cite{samari2024single} does not incorporate input constraints, \cite{nejati2022data} addresses them; however, this results in \emph{bilinearity} in the safety controller design, thus increasing computational complexity. Our approach, in contrast, supports input constraints without introducing any bilinearity.
	
	\textbf{Primary contributions.} Motivated by the challenges discussed, this work develops a direct data-driven framework to synthesize CBCs and corresponding safety controllers from collected data. Our framework is specifically designed for discrete-time systems with \emph{general nonlinear dynamics and input constraints}, ensuring a safety certificate over finite time horizons. To this end, we first employ an innovative control technique known as ``adding one integrator'', where we treat the control input as a state variable and incorporate its dynamics into those of the system, thereby constructing an \emph{augmented system} with a virtual control input. This step is then completed by adjusting the admissible control input set and embedding it within the augmented system's state set, ultimately enforcing the input constraints. Consequently, we design a virtual static-feedback safety controller for the augmented system, resulting in a \emph{dynamic} safety controller for the actual system. 
	
	To achieve this, we collect input-state data from the augmented system and leverage it to design a CBC and its associated static safety controller for the augmented system, ensuring the finite-time safety of the actual system while satisfying its input constraints. Our proposed conditions enable the design of CBCs and safety controllers for the augmented system using straightforward, data-dependent semidefinite programs. To maximize safety guarantees—ensuring the safety property over an extended finite time horizon—we also propose an alternative computational approach using a scenario convex program, allowing us to achieve the maximum horizon guarantee within our approach. Finally, we demonstrate the effectiveness of our data-driven framework through two complex case studies involving  highly nonlinear terms with unknown dynamics.
	
	\textbf{Organization.} The rest of the paper is structured as follows. In Section \reff{Section: Problem Formulation}, we present the mathematical preliminaries, notations, and formal definitions of discrete-time general nonlinear systems, CBCs, the augmented system, and the augmented version of CBCs. Section \reff{Section: Data-driven-disc} is allocated to presenting how we collect data and introducing our main results, guaranteeing system safety. In Section \reff{Section: bound}, we provide an alternative computational approach that can potentially boost the safety guarantee. Simulation results are provided in Section \reff{Section: Simul}, and the paper is concluded in Section \reff{Section: Conc}.
	
	\section{General Nonlinear Framework} \label{Section: Problem Formulation}
	
	\subsection{Notation}
	We adopt the following notation throughout the paper. We denote the set of real numbers as $\R$, while $\Rp$ and $\Rpz$ represent the sets of positive and non-negative real numbers, respectively. The set of non-negative integers is written as $\N$, while $\Np$ refers to the set of positive integers.
	Moreover, $\I_n$ denotes the $n \times n$ identity matrix, while $\pmb{0}_{n \times m}$ and $\pmb{0}_{n}$ represent the $n \times m$ zero matrix and the zero vector of dimension $n$, respectively. For any matrix $A$, $A^\top$ denotes its transpose. The symbol $\star$ is used to indicate the transposed element in the symmetric position within a symmetric matrix.
	We also adopt the notations $P \succ 0$ and $P \succeq 0$ to indicate that the \emph{symmetric} matrix $P$ is positive definite and positive semi-definite, respectively. Given a square matrix $P$, the maximum and minimum eigenvalues of $P$ are represented by $\lambda_{\max}(P)$ and $\lambda_{\min}(P)$, respectively.
	The Euclidean norm of a vector $x \in \R^{n}$ is denoted by $\vert x \vert$, while for a given matrix $A$, $\Vert A \Vert$ represents its induced 2-norm.
	The horizontal concatenation of vectors $x_i \in \R^n$ to form an $n \times N$ matrix is expressed as $\begin{bmatrix} x_1 & \hspace{-0.2cm} x_2 & \hspace{-0.2cm} \dots & \hspace{-0.2cm} x_N \end{bmatrix}$.
	Furthermore, the Cartesian product of sets $X$ and $Y$ is denoted by $X \times Y$, the relative complement of $X$ and $Y$ by $X \backslash Y$, and the union of $X$ and $Y$ by $X \cup Y$.
	
	\subsection{Discrete-Time General Nonlinear Systems}
	We start with a description of the discrete-time general nonlinear system that is of interest in this work.
	
	\begin{deff}[{\normalfont \textbf{dt-GNS}}]\label{def: dt-GNS}
		A discrete-time general nonlinear system (dt-GNS) $\Upsilon$ evolves according to
		\begin{align}
			\Upsilon : x^+ = \digamma(x, u), \label{eq: gen_disc_org}
		\end{align}
		where $x^+$ represents the state variables at the next time step, i.e., $x^+ \coloneq x(k + 1), \; k \in \N$. The state $x \in \mathcal{X}$, where $\mathcal{X} \subset \R^n$ is the state set, and the control input $u \in U$, where
		\begin{align}
			U = \big\{u \in \R^m  ~\big \vert~ \mathcal{C}_j^\top u \leq 1, \; j = 1, \dots, \mathrm{j}\big\} \subset \R^m, \label{eq: input_const}
		\end{align}
		is the set of admissible control input set, with each $\mathcal{C}_j \in \R^m$ being a known vector that imposes a constraint on the control input. The mapping $\digamma : \R^n \times \R^m \rightarrow \R^n$ is smooth and devoid of any constant term. For simplicity, and without loss of generality, we can rewrite system \eqref{eq: gen_disc_org} as
		\begin{align}
			\Upsilon : x^+ = A f(x, u), \label{eq: gen_disc_rewritten}
		\end{align}
		where $A \in \R^{n \times N}$ is a constant matrix, and $f : \R^n \times \R^m \rightarrow \R^N$ is a smooth mapping without any constant term. The state trajectory of system $\Upsilon$ at time step $k \in \N$, given the initial condition $x_0 = x(0)$ and input signal $u(\cdot)$, is denoted by $x_{x_0u}(k)$.
		We use the tuple $\Upsilon = (A, f, \mathcal{X}, U)$ to represent the dt-GNS in \eqref{eq: gen_disc_rewritten}. 
	\end{deff}
	
	In this paper, we assume that matrix $A$ is \emph{unknown}, which aligns with the reality in many practical scenarios. However, a dictionary for $f(x, u)$—referring to a \emph{library or family} of functions—is assumed to be available that best fits the actual dynamics by being \emph{as extensive as necessary}.
	Specifically, the dictionary for $f(x, u)$ is comprehensive enough to include all potential terms in the true system's dynamics, even allowing for the inclusion of irrelevant terms. This dictionary comprises all linear functions of the state and control input, along with nonlinear terms derived from system insights, \emph{i.e.,}
	\begin{align}
		f(x, u) \coloneq \begin{bmatrix}
			x\\
			u\vspace{-.225cm}\\
			\tikz\draw [thick,dashed] (0,0) -- (1.25,0);\\
			\pmb{\Psi}(x, u)
		\end{bmatrix}\!, \label{eq: dictionary}
	\end{align}
	where $\pmb{\Psi} : \R^n \times \R^m \rightarrow \R^{N-(n+m)}$ represents the nonlinear functions.

    \begin{remark}[\normalfont \textbf{Dictionary Availability}]
        The assumption of an extensive dictionary for $f(x, u)$ is common in the data-driven literature (\emph{e.g.,} \cite{de2023learning, nejati2022data, samari2024single}) and is often satisfied in various systems, such as electrical and mechanical ones. In such cases, knowledge of the dynamics (i.e., an exaggerated dictionary for $f(x, u)$) is derived from first principles, even though the exact system parameters (i.e., matrix $A$) may still remain unknown.
    \end{remark}
	
	The primary objective of this work is to provide a safety certificate for the unknown dt-GNS, described in \eqref{eq: gen_disc_rewritten}, by employing the concept of control barrier certificates. Accordingly, we formally introduce this concept in the next subsection.
	
	\subsection{Control Barrier Certificates}
	Having described the dt-GNS as per Definition \reff{def: dt-GNS}, we now move on to present the notion of control barrier certificates for such a system, formalized in the subsequent definition \cite{prajna2004safety}.
	
	\begin{deff}[\textbf{CBC}] \label{def: CBC}
		Consider the dt-GNS $\Upsilon = (A, f, \mathcal{X}, U)$, with the initial set $\Xo \subset \mathcal{X}$ and the unsafe set $\Xu \subset \mathcal{X}$. A smooth function $\pmb{\mathds{B}}: \mathcal{X} \rightarrow \Rpz$ is defined as a control barrier certificate (CBC) for the dt-GNS over a time horizon $[0, T-1]$, with $T \in \Np$,  if there exist $\eta, \gamma \in \Rp$, with $\eta + cT < \gamma$, and $c \in \Rpz$, such that
		\begin{subequations}
			\begin{itemize}
				\item $\forall x \in \Xo$:
				\begin{align}
					\pmb{\mathds{B}}(x) \leq \eta, \label{eq: con_bar_1}
				\end{align}
				\item $\forall x \in \Xu$:
				\begin{align}
					\pmb{\mathds{B}}(x) \geq \gamma,  \label{eq: con_bar_2}
				\end{align}
				\item $\forall x \in \mathcal{X}, \; \exists u \in U,$ such that
				\begin{align}
					\pmb{\mathds{B}}(A f(x, u)) \leq \pmb{\mathds{B}}(x) + c.  \label{eq: con_bar_3}
				\end{align}
			\end{itemize}
		\end{subequations}
	\end{deff}
	
	As highlighted in Definition \reff{def: CBC}, CBCs are defined over the state space of the system, imposing a series of inequalities on itself, \emph{i.e.,} conditions \eqref{eq: con_bar_1} and \eqref{eq: con_bar_2}, as well as on its one-step transition, \emph{i.e.,} condition \eqref{eq: con_bar_3}. Then, a suitable level set of the CBC, denoted as $\eta$, acts as a separator, distinguishing the unsafe set (\emph{i.e.,} $\Xu$) from the trajectories of the system originating from the initial state set (\emph{i.e.,} $\Xo$).
	
	\begin{remark}[\normalfont \textbf{Special Case $c \rightarrow 0$}]
		Inspired by $c$-martingale properties in the stochastic setting~\cite{nejati2023formal}, we allow the CBC in \eqref{eq: con_bar_3} to decay with respect to a constant $c$. This approach significantly increases the likelihood of finding a feasible CBC while providing safety guarantees over finite time horizons. Since we enforce $T < \frac{\gamma - \eta}{c}$ in Definition \reff{def: CBC}, it can be deduced that the time horizon $T$ converges to infinity as $c$ approaches zero.
	\end{remark}
	
	\begin{remark}[\normalfont \textbf{Empty Intersection between $\Xo$ and $\Xu$}]\label{remark 1}
		Of note is that $\Xo$ and $\Xu$ should be disjoint to guarantee the safety property described in Definition \reff{def: CBC}. Specifically, we impose the condition $\gamma > \eta + cT$, leading to $\gamma > \eta$, which implicitly indicates that no intersection exists between the sets $\Xo$ and $\Xu$, as derived from conditions \eqref{eq: con_bar_1} and \eqref{eq: con_bar_2}.
	\end{remark}
	
	The following theorem, borrowed from \cite{prajna2004safety,nejati2023formal}, utilizes the concept of CBCs and ensures finite-horizon safety certificates for the dt-GNS $\Upsilon$.
	
	\begin{thm}[\normalfont \textbf{Safety Guarantee for dt-GNS}] \label{thm: model-based}
		Consider the dt-GNS $\Upsilon = (A, f, \mathcal{X}, U)$, as described in Definition \reff{def: dt-GNS}, where $\Xo$ and $\Xu$ represent the initial and unsafe sets, respectively. Let $\pmb{\mathds{B}}$ be a CBC for $\Upsilon$ as defined in Definition \reff{def: CBC}. Then, the dt-GNS is considered $T$-horizon safe, implying that the system's state trajectories will not enter the unsafe set $\Xu$ initiated from the initial set $\Xo$, i.e., $x_{x_0 u} \notin \Xu$ for any $x_0 \in \Xo$ and any $k \in [0, T-1]$, with $T < \frac{\gamma - \eta}{c}$, under the control input $u(\cdot)$ being designed to satisfy condition \eqref{eq: con_bar_3}.
	\end{thm}
	
	Given that the dt-GNS in this work is in the general form of nonlinear systems—potentially introducing nonlinearities in the control input itself—designing a (static) feedback control law to ensure the $T$-horizon safety property is challenging. Additionally, one of the objectives is to enforce the input constraints in \eqref{eq: input_const} while simultaneously ensuring the $T$-horizon safety of the dt-GNS. These challenges motivate us to treat the control input as a state variable, incorporating its dynamics into the system’s dynamics to construct an augmented system, where its state includes both the system's actual states and the control input (cf. system \eqref{eq: final_sys}). Consequently, instead of designing a static-feedback control law to satisfy the safety property, we design a \emph{dynamic controller} for the actual system.
	
	Notably, our proposed approach offers advantages over existing methods, \emph{e.g.,} \cite{nejati2022data}, where bilinearity arises in constructing the CBC and its associated safety controller under certain control input constraints (cf. condition (20) and Remark 13 in \cite{nejati2022data})—a challenge not present in our setting. It is worth noting that the technique of ``adding one integrator'' has been widely used in the sliding mode control literature (see, for instance, \cite{bartolini1998chattering}). We introduce the augmented system in the following subsection.
	
	\begin{figure*}
		\centering
		
		\subfloat[The actual admissible control input set]{%
			\includegraphics[width=0.6\textwidth]{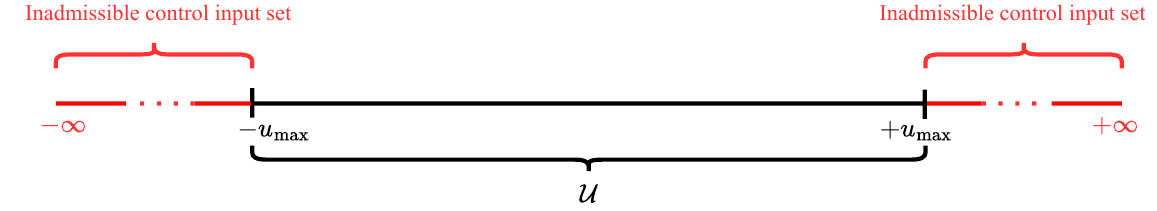} 
			\label{fig:subfig1}
		}

		\subfloat[The adjustments on the actual admissible control input set to obtain sets \eqref{eq1}-\eqref{eq2}]{%
			\includegraphics[width=0.6\textwidth]{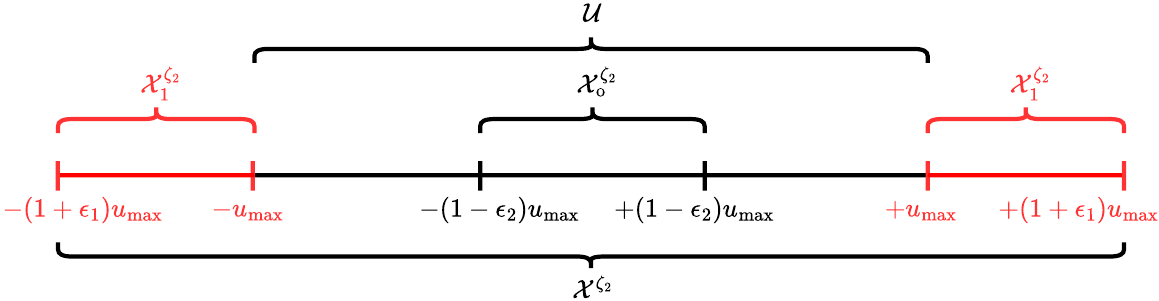} 
			\label{fig:subfig2}
		}
		
		\caption{As shown in Subfigure \reff{fig:subfig1}, if one selects the entire admissible control input set as the initial set for $\zeta_2$ (\emph{i.e., } $\mathcal{X}^{\zeta_2}_\mathrm{o}$), it would be impossible to find a level set that distinguishes the inadmissible set (\emph{i.e.,} the unsafe set for $\zeta_2$ denoted by $\mathcal{X}^{\zeta_2}_\mathrm{1}$) from $\mathcal{X}^{\zeta_2}_\mathrm{o}$. However, as illustrated in Subfigure \reff{fig:subfig2}, with the proposed adjustments, it becomes feasible to find a level set that effectively separating $\mathcal{X}^{\zeta_2}_\mathrm{o}$ from $\mathcal{X}^{\zeta_2}_\mathrm{1}$\!.
		}
		\label{fig:mainfig}
	\end{figure*}
	
	\subsection{Augmented dt-GNS} \label{subsec:a}
	Being motivated to incorporate the control input dynamic into the dt-GNS dynamic, we now proceed with introducing the augmented system. 
	To do so, we first define $\zeta_1 \coloneq x$ and $\zeta_2 \coloneq u$, and introduce the augmented state
	\begin{align*}
		\zeta \coloneq \begin{bmatrix}
			\zeta_1\\
			\zeta_2
		\end{bmatrix} \in \R^{n+m},
	\end{align*}
	while redefining the dictionary \eqref{eq: dictionary} as
	\begin{align}
		\mathcal{F}(\zeta) \coloneq \begin{bmatrix}
			\zeta\vspace{-.225cm}\\
			\tikz\draw [thick,dashed] (0,0) -- (1.25,0);\\
			\pmb{\Psi}(\zeta)
		\end{bmatrix}\!, \label{eq: final_dict}
	\end{align}
	where $\pmb{\Psi}(\zeta)$ represents the nonlinear terms.
	
	We consider the new ``virtual control input'' $\vartheta \in \mathcal{U}^\zeta$, where $\mathcal{U}^\zeta \subset \R^m$ is the virtual control input set, and add the controller dynamic to the augmented system in the form of $u^+ = \vartheta$, \emph{i.e.,} $\zeta_2^+ = \vartheta$. 
	Consequently, the augmented dt-GNS $\Upsilon_a$, denoted by A-dt-GNS, is formulated as
	\begin{align}
		\Upsilon_a : \zeta^+ = \mathcal{A} \mathcal{F}(\zeta) + \mathcal{B}\vartheta,\label{eq: final_sys}
	\end{align}
	where
	\begin{align*}
		\mathcal{A} = \begin{bmatrix}
			A\\
			\pmb{0}_{m \times N}
		\end{bmatrix}\in\R^{(n+m) \times N}, \quad \mathcal{B} = \begin{bmatrix}
		\pmb{0}_{n \times m}\\
		\I_m
		\end{bmatrix} \in \R^{(n+m) \times m}.
	\end{align*}
	
	Since the state of the A-dt-GNS $\Upsilon_a$ should stay safe in our setting, we now define the state, initial, and unsafe sets for the A-dt-GNS $\Upsilon_a$. Given that these sets for $\zeta_1$ are exactly similar to those of the state $x$ for the actual dt-GNS $\Upsilon$, we define
	\begin{subequations}\label{eq: sets_z1}
		\begin{align}
			\mathcal{X}^{\zeta_1} & \coloneq \mathcal{X},\\
			\mathcal{X}^{\zeta_1}_\mathrm{o} & \coloneq \Xo,\\
			\mathcal{X}^{\zeta_1}_\mathrm{1} & \coloneq \Xu.
		\end{align}
	\end{subequations}
	
	For the sets related to $\zeta_2$, we aim to construct them by making specific adjustments to the set $U$ in \eqref{eq: input_const}. Specifically, one cannot naively define $\mathcal{X}^{\zeta_2}_\mathrm{o} \coloneq U$ and $\mathcal{X}^{\zeta_2}_\mathrm{1} \coloneq \R^m \backslash U$, as this would result in shared boundaries on the edges, making it impossible to find a level set that effectively separates these sets 
	(cf. Remark \reff{remark 1}). Therefore, we should select a subset of \eqref{eq: input_const} as the initial set for $\zeta_2$ and define its unsafe set starting from the boundaries of \eqref{eq: input_const} onward.
	
	To achieve this, we introduce the design parameters $\epsilon_1, \epsilon_2 \in (0, 1)$, and define the following sets:
	\begin{subequations}\label{eq: sets_z2}
		\begin{align}
			\mathcal{X}^{\zeta_2} & \coloneq \big\{\zeta_2 \in \R^m  ~\big \vert ~\mathcal{C}_j^\top \zeta_2 \leq 1 + \epsilon_1, \; j = 1, \dots, \mathrm{j}\big\} \subset \R^m, \label{eq1}\\
			\mathcal{X}^{\zeta_2}_\mathrm{o} & \coloneq \big\{\zeta_2 \in \R^m  ~\big \vert ~ \mathcal{C}_j^\top \zeta_2 \leq 1 - \epsilon_2, \; j = 1, \dots, \mathrm{j}\big\} \subset \R^m,\\
			\mathcal{U} & \coloneq \big\{\zeta_2 \in \R^m ~\big \vert ~ \mathcal{C}_j^\top \zeta_2 \leq 1, \; j = 1, \dots, \mathrm{j}\big\} \subset \R^m,\label{eq:tmp}\\
			\mathcal{X}^{\zeta_2}_\mathrm{1} & \coloneq (\mathcal{X}^{\zeta_2} \backslash \mathcal{U}) \cup \partial \mathcal{U} \subset \R^m, \label{eq2}
		\end{align}
	\end{subequations}
	where \eqref{eq:tmp} is exactly the same as \eqref{eq: input_const} since $\zeta_2 = u$, and $\partial \mathcal{U} \coloneq  \big\{\zeta_2 \in \R^m  ~\big \vert ~ \mathcal{C}_j^\top \zeta_2 = 1 , \; j = 1, \dots, \mathrm{j}\big\} \subset \R^m$ is the \emph{boundary} of the set $\mathcal{U}$. 
	
	To provide the general reader with the necessary intuition regarding the definition of various sets for $\zeta_2$ within our framework, Figure \reff{fig:mainfig} offers a visual representation of how we define sets \eqref{eq1}-\eqref{eq2} based on the admissible control input set \eqref{eq: input_const}. In this illustrative example, we assume that the single control input is constrained by $-u_{\max} \leq u \leq u_{\max}$, where $u_{\max} \in \Rp$ is an arbitrary value.
	
	We now consider \eqref{eq: sets_z1}, \eqref{eq: sets_z2}, and define the state set of the A-dt-GNS $\Upsilon_a $ as $\mathcal{X}^\zeta \coloneq \mathcal{X}^{\zeta_1} \times \mathcal{X}^{\zeta_2} \subset \R^{n+m}$, the initial set of the A-dt-GNS as $\mathcal{X}^\zeta_{\mathrm{o}} \coloneq \mathcal{X}^{\zeta_1}_\mathrm{o} \times \mathcal{X}^{\zeta_2}_\mathrm{o} \subset \mathcal{X}^\zeta$, and the unsafe set of the A-dt-GNS as 
	\begin{align*}
			\mathcal{X}^\zeta_{\mathrm{1}} \coloneq \overbrace{(\mathcal{X}^{\zeta_1} \times \mathcal{X}^{\zeta_2}_\mathrm{1})}^{(a)} \cup \overbrace{(\mathcal{X}^{\zeta_1}_\mathrm{1} \times \mathcal{X}^{\zeta_2})}^{(b)} \subset \mathcal{X}^\zeta,
		\end{align*} where the term $(a)$ is to ensure that the input constraints are met, while the term $(b)$ is to ensure the safety of the system's states (cf. Figure \reff{fig:states} in Section \reff{Section: Simul}). We denote the state trajectory of $\Upsilon_a$ by $\zeta_{\zeta_0\vartheta}(k)$ at time step $k \in \N$, starting from the initial condition $\zeta_0 = \zeta(0)$ and under the virtual input signal $\vartheta(\cdot)$.
	We represent the  A-dt-GNS \eqref{eq: final_sys} using the tuple $\Upsilon_a = (\mathcal{A}, \mathcal{B},  \mathcal{F}, \mathcal{X}^\zeta, \mathcal{U}^\zeta)$.
	
	\begin{remark}[\normalfont \textbf{On Choosing $\epsilon_1$ and $\epsilon_2$}]\label{remark 2}
		It is desirable to choose $\epsilon_1$ as close to zero as possible, since this can facilitate satisfying conditions \eqref{eq: A-con_bar_3} and \eqref{eq: con_thm_con4} with potentially a smaller level set, leading to a longer time-horizon guarantee. Conversely, for $\epsilon_2$, it is beneficial to set its value as close to one as possible; otherwise, finding valid level sets with an adequate guarantee may become challenging.
	\end{remark}
	
	Having introduced the A-dt-GNS $\Upsilon_a$ with its corresponding details, we now proceed with defining the augmented CBC for this system, as formalized in the following definition.
	
	\begin{deff}[\textbf{A-CBC}] \label{def: A-CBC}
		Consider the A-dt-GNS $\Upsilon_a \!=\!  (\mathcal{A}, \mathcal{B}, \mathcal{F},\mathcal{X}^\zeta,\\ \mathcal{U}^\zeta)$, with the initial and unsafe sets $\mathcal{X}^\zeta_{\mathrm{o}}, \mathcal{X}^\zeta_{\mathrm{1}} \subset \mathcal{X}^\zeta$, respectively. A smooth function $\pmb{\mathds{B}_a}: \mathcal{X}^\zeta \rightarrow \Rpz$ is defined as an augmented CBC (A-CBC) for the A-dt-GNS over a time horizon $[0, T-1]$, with $T \in \Np$, if there exist $\eta_a, \gamma_a \in \Rp$, with $\eta_a + c_aT < \gamma_a$, and $c_a \in \Rpz$, such that
		\begin{subequations}
			\begin{itemize}
				\item $\forall \zeta \in \mathcal{X}^\zeta_{\mathrm{o}}$:
				\begin{align}
					\pmb{\mathds{B}_a}(\zeta) \leq \eta_a, \label{eq: A-con_bar_1}
				\end{align}
				\item $\forall \zeta \in \mathcal{X}^\zeta_{\mathrm{1}}$:
				\begin{align}
					\pmb{\mathds{B}_a}(\zeta) \geq \gamma_a,  \label{eq: A-con_bar_2}
				\end{align}
				\item $\forall \zeta \in \mathcal{X}^\zeta, \; \exists \vartheta \in \mathcal{U}^\zeta,$ such that
				\begin{align}
					\pmb{\mathds{B}_a}(\mathcal{A} \mathcal{F}(\zeta) + \mathcal{B}\vartheta) \leq \pmb{\mathds{B}_a}(\zeta) + c_a.  \label{eq: A-con_bar_3}
				\end{align}
			\end{itemize}
		\end{subequations}
	\end{deff}
	
	Inspired by Theorem \reff{thm: model-based}, we now provide the next result on the $T$-horizon safety of the A-dt-GNS $\Upsilon_a$ to ensure that the state trajectory of the augmented system stays safe over the time horizon $[0, T-1]$, with $T < \frac{\gamma_a - \eta_a}{c_a}$.
	
	\begin{coro}[\normalfont\textbf{Safety Guarantee for A-dt-GNS}]\label{Corollary 1}
		Consider the A-dt-GNS $\Upsilon_a =  (\mathcal{A}, \mathcal{B}, \mathcal{F},\mathcal{X}^\zeta, \mathcal{U}^\zeta)$, as described in \eqref{eq: final_sys}, with $\mathcal{X}^\zeta_{\mathrm{o}}$ and $\mathcal{X}^\zeta_{\mathrm{1}}$ being its initial and unsafe sets, respectively. Let $\pmb{\mathds{B}_a}$ be an A-CBC for $\Upsilon_a$ as defined in Definition \reff{def: A-CBC}. Then, the A-dt-GNS is considered $T$-horizon safe, implying that the system's state trajectories will not enter the unsafe set $\mathcal{X}^\zeta_{\mathrm{1}}$ initiated from the initial set $\mathcal{X}^\zeta_{\mathrm{o}}$, i.e., $\zeta_{\zeta_0 \vartheta} \notin \mathcal{X}^\zeta_{\mathrm{1}}$ for any $\zeta_0 \in \mathcal{X}^\zeta_{\mathrm{o}}$ and any $k \in [0, T-1]$, with $T < \frac{\gamma_a - \eta_a}{c_a}$, under the virtual control input $\vartheta(\cdot)$  being designed to satisfy condition \eqref{eq: A-con_bar_3}.
	\end{coro}
	
	\begin{proof}
		Assume that there exists an $\pmb{\mathds{B}_a}(\zeta(\cdot))$ serving as an A-CBC for the A-dt-GNS $\Upsilon_a$. Then, upon condition \eqref{eq: A-con_bar_3}, we have $\pmb{\mathds{B}_a}(\zeta(k+1)) - \pmb{\mathds{B}_a}(\zeta(k)) \leq c_a$, with $\pmb{\mathds{B}_a}(\zeta(k+1)) =\pmb{\mathds{B}_a}(\mathcal{A} \mathcal{F}(\zeta) + \mathcal{B}\vartheta)$, leading to $\pmb{\mathds{B}_a}(\zeta(k)) - \pmb{\mathds{B}_a}(\zeta(0)) \leq kc_a$, recursively. Based on \eqref{eq: A-con_bar_1}, we know that $\pmb{\mathds{B}_a}(\zeta(0)) \leq \eta_a$, and thus, one has $\pmb{\mathds{B}_a}(\zeta(k)) \leq kc_a + \eta_a$. Having known that $\eta_a + c_aT < \gamma_a$, we clearly have $\pmb{\mathds{B}_a}(\zeta(k)) < \gamma_a$. Hence, according to \eqref{eq: A-con_bar_2}, it can be deduced that $\zeta(k) \notin \mathcal{X}^\zeta_{\mathrm{1}}$ for any $k \in [0, T-1]$, concluding the proof. It is evident that if $c_a = 0$, the safety guarantee extends to an infinite time horizon.
	\end{proof}
	
	\begin{remark}[\normalfont \textbf{Safety Guarantee and Input Constraint Enforcement for dt-GNS}]\label{Re_New}
		Given that the first state of the A-dt-GNS (i.e., $\zeta_1$) is exactly the main state of the actual system (i.e., $x$), it is evident that enforcing the safety property over the augmented system according to Corollary~\reff{Corollary 1} will ensure the safety property for the actual system as well. This is also apparent from the definition of the regions of interest for $\zeta_1$ in \eqref{eq: sets_z1}, which match exactly those for the actual system. Additionally, since the second state of the A-dt-GNS (i.e., $\zeta_2$) corresponds to the control input of the actual system (i.e., $u$), it follows that satisfying the safety property for this state, according to the defined regions of interest for $\zeta_2$ in \eqref{eq: sets_z2}, will enforce the input constraint $u$ in the actual system. It is evident that since $\zeta_2 = u$ and $\zeta_2^+ = \vartheta$, we have $u^+ = \vartheta$, indicating that the safety controller for the actual system is dynamic.
	\end{remark}
	
	Establishing an A-CBC for the A-dt-GNS $\Upsilon_a$ to ensure system safety, as outlined in Corollary \reff{Corollary 1}, requires knowledge of the dynamic upon which the states evolve, \emph{i.e.,} $\mathcal{A} \mathcal{F}(\zeta) + \mathcal{B}\vartheta$, as it appears in condition \eqref{eq: A-con_bar_3}. However, this information is not readily available in many practical applications. With this critical challenge in mind, we are now prepared to formally present the problem addressed in this work.
	
	\begin{resp}
		\begin{problem} \label{prob 1}
			Consider the A-dt-GNS $\Upsilon_a =  (\mathcal{A}, \mathcal{B}, \mathcal{F},\mathcal{X}^\zeta, \mathcal{U}^\zeta)$ in \eqref{eq: final_sys}, where matrix $\mathcal{A}$ is unknown. Using only a single input-state trajectory collected from $\Upsilon_a$, develop a formal data-driven framework for designing an A-CBC and its associated safety controller, as defined in Definition \reff{def: A-CBC}, to ensure the safety property guaranteed by Corollary \reff{Corollary 1} without explicitly identifying the system model. Consequently, the dt-GNS $\Upsilon$ remains safe, and the input constraints \eqref{eq: input_const} would be enforced as outlined in Remark~\reff{Re_New}.
		\end{problem}
	\end{resp}
	
	With the main problem of the paper established, we move forward to present our proposed data-driven framework in the next section.
	
	\section{Data-Driven Design of A-CBC and Safety Controller}\label{Section: Data-driven-disc}
	This section presents our data-driven scheme for constructing an A-CBC and its associated safety controller for the unknown A-dt-GNS, described in \eqref{eq: final_sys}.
	To achieve this, we first consider the A-CBC in the form of a quadratic function, expressed as $\pmb{\mathds{B}_a}(\zeta) = \zeta^\top \mathcal{P} \zeta$, where $\mathcal{P} \succ 0$. Following this, we collect input-state data (\emph{i.e.,} observations) over the time interval $\left[0, \mathcal{T}\right]$, where $\mathcal{T} \in \Np$ represents the total number of collected samples:
	\begin{subequations}\label{eq: data}
		\begin{align}
			\pmb{\mathcal{S}} & \coloneq \begin{bmatrix}
				\zeta(0) & \zeta(1) & \dots & \zeta(\mathcal{T} - 1)
			\end{bmatrix}\!,\\
			\pmb{\mathcal{I}} &  \coloneq \begin{bmatrix}
				\vartheta(0) & \vartheta(1) & \dots & \vartheta(\mathcal{T} - 1)
			\end{bmatrix}\!,\\
			\pmb{\mathcal{S}^+} & \coloneq \begin{bmatrix}
				\zeta(1) & \zeta(2) & \dots & \zeta(\mathcal{T})
			\end{bmatrix}\!.
		\end{align}
	\end{subequations}
	We emphasize that the set of input-state trajectories in \eqref{eq: data} is regarded as a \emph{single trajectory} of the unknown A-dt-GNS $\Upsilon_a$.
	
	\begin{remark}[\normalfont\textbf{On Collecting $\pmb{\mathcal{S}^+}$}]
		It is important to underline that in the discrete-time framework, the trajectory of $\pmb{\mathcal{S}^+}$ can be accurately captured, as the system's unknown dynamics progress recursively at discrete time steps. This is advantageous compared to the continuous-time scenario (e.g., \cite{nejati2022data}), where $\pmb{\mathcal{S}^+}$ involves state derivatives at sampling instants, which are typically unavailable as direct observations.
	\end{remark}
	
	Having been inspired by~\cite{de2023learning}, we now provide the subsequent lemma upon whose conditions we are able to characterize the data-based representation of the A-dt-GNS $\Upsilon_a$.
	
	\begin{lem}[\normalfont\textbf{Data-based {\normalfont \textbf{A-dt-GNS}}}]\label{lemma 1}
		Suppose $\mathcal{Z} \in \R^{\mathcal{T} \times N}$ is an arbitrary matrix being partitioned into $\mathcal{Z}_1 \in \R^{\mathcal{T} \times (n+m)}$ and $\mathcal{Z}_2 \in \R^{\mathcal{T} \times (N - (n+m))}$ such that one has
		\begin{align}
			\I_{N} = \pmb{\mathcal{M}} \overbrace{ \begin{bmatrix}
					\mathcal{Z}_1 & \mathcal{Z}_2
			\end{bmatrix}}^{\mathcal{Z}}, \label{eq: Lemma-con}
		\end{align}
		where $\pmb{\mathcal{M}} \in \R^{N \times \mathcal{T}}$ is a full row-rank matrix being obtained from the dictionary \eqref{eq: final_dict} and samples $\pmb{\mathcal{S}}$ as
		\begin{align*}
			\pmb{\mathcal{M}} \coloneq \begin{bmatrix}
				\zeta(0) & \zeta(1) & \dots & \zeta(\mathcal{T} - 1)\\
				\pmb{\Psi}(\zeta(0)) & \pmb{\Psi}(\zeta(1)) & \dots & \pmb{\Psi}(\zeta(\mathcal{T} - 1))
			\end{bmatrix}\!.
		\end{align*}
		Under the control law
		\begin{align}
			\vartheta = \mathcal{K}\mathcal{F}(\zeta), \label{eq: Lemma-control}
		\end{align}
		where $\mathcal{K} = \pmb{ \mathcal{I}} \mathcal{Z}$, the closed-loop  data-based representation of A-dt-GNS $\Upsilon_a$ is
		\begin{align}
			\zeta^+ = \pmb{\mathcal{S}^+} \mathcal{Z}_1 \zeta + \pmb{\mathcal{S}^+} \mathcal{Z}_2 \pmb{\Psi}(\zeta). \label{eq: Lemma-closed}
		\end{align}
	\end{lem}
	
	\begin{proof}
		We first start with the fact that the data collected should fit the A-dt-GNS, \emph{i.e.,} $\zeta^+ = \mathcal{A} \mathcal{F}(\zeta) + \mathcal{B}\vartheta,$ as
		\begin{align}
			\pmb{\mathcal{S}^+} = \mathcal{A} \pmb{\mathcal{M}} + \mathcal{B} \pmb{\mathcal{I}}. \label{eq: Lemma-tmp1}
		\end{align}
		Furthermore, we have
		\begin{align*}
			\zeta^+ & = \mathcal{A} \mathcal{F}(\zeta) + \mathcal{B}\vartheta\\
			& \!\! \overset{\eqref{eq: Lemma-control}}{=}  \mathcal{A} \mathcal{F}(\zeta) + \mathcal{B} \pmb{\mathcal{I}}\mathcal{Z}\mathcal{F}(\zeta)= (\mathcal{A} + \mathcal{B} \pmb{\mathcal{I}}\mathcal{Z}) \mathcal{F}(\zeta)\\
			& \!\! \overset{\eqref{eq: Lemma-con}}{=} (\mathcal{A} \overbrace{\pmb{\mathcal{M}} \mathcal{Z}}^{\I_N} + \mathcal{B} \pmb{\mathcal{I}}\mathcal{Z}) \mathcal{F}(\zeta)= (\mathcal{A} \pmb{\mathcal{M}} + \mathcal{B} \pmb{\mathcal{I}}) \mathcal{Z} \mathcal{F}(\zeta)\\
			& \!\! \overset{\eqref{eq: Lemma-tmp1}}{=} \pmb{\mathcal{S}^+} \mathcal{Z}\mathcal{F}(\zeta) \overset{\eqref{eq: final_dict}}{=} \pmb{\mathcal{S}^+} \begin{bmatrix}
				\mathcal{Z}_1 & \mathcal{Z}_2
			\end{bmatrix}   \begin{bmatrix}
				\zeta\vspace{-.225cm}\\
				\tikz\draw [thick,dashed] (0,0) -- (1.25,0);\\
				\pmb{\Psi}(\zeta)
			\end{bmatrix}\\
			& = \pmb{ \mathcal{S}^+} \mathcal{Z}_1 \zeta + \pmb{\mathcal{S}^+} \mathcal{Z}_2 \pmb{\Psi}(\zeta),
		\end{align*}
		which concludes the proof.
	\end{proof}
	
	\begin{remark}[\normalfont\textbf{Richness of Collected Data}]\label{rem:tmp}
		It is evident that a necessary condition for \eqref{eq: Lemma-con} to be satisfied is that the matrix $\pmb{\mathcal{M}}$ must possess full row rank. This requires collecting at least $N + 1$ data points, i.e., $\mathcal{T} \geq N + 1$, which can be readily verified during data collection. This requirement can be interpreted as a natural extension of the rank condition imposed on $\pmb{\mathcal{S}}$ in the context of linear systems \cite{de2019formulas}.
	\end{remark}
	
	Since our safety analysis is confined to the compact set $\mathcal{X}^\zeta$, we choose to establish an upper bound for the nonlinear part embedded in $\pmb{\mathds{B}_a}(\mathcal{A} \mathcal{F}(\zeta) + \mathcal{B}\vartheta)$ (cf. condition \eqref{eq: con_thm_con4}). Specifically, on a compact set $\mathcal{X}^\zeta$, where $\pmb{\Psi}(\zeta)$ consists of smooth nonlinear functions, computing this upper bound is always feasible as such functions reach their maximum within the compact set and exhibit no unbounded behavior. By employing this upper bound, which is always achievable in our proposed framework, we are able to design A-CBCs for general nonlinear systems. This approach also yields data-dependent \emph{semidefinite programs}, offering greater scalability compared to existing methods for nonlinear polynomial systems that rely on data-dependent SOS programs.
	
	Under the closed-loop data-driven representation of the A-dt-GNS derived in Lemma \reff{lemma 1}, we now proceed with proposing the following theorem, which constitutes the main contribution of this work. This result enables the design of an A-CBC and its associated safety controller using only a single input-state trajectory from the A-dt-GNS.
	
	\begin{thm}[\normalfont\textbf{Data-Driven A-CBC and Safety Controller}]\label{thm: main}
		Consider the unknown A-dt-GNS $\Upsilon_a$ as per \eqref{eq: final_sys} with its closed-loop data-based representation $\mathcal{A} \mathcal{F}(\zeta) + \mathcal{B}\vartheta = \pmb{ \mathcal{S}^+} \mathcal{Z}_1 \zeta + \pmb{\mathcal{S}^+} \mathcal{Z}_2 \pmb{\Psi}(\zeta)$ as in Lemma \reff{lemma 1}. Let there exist the matrix $\mathcal{Z}_2 \in \R^{\mathcal{T} \times (N - (n+m))}$ such that
		\begin{subequations}\label{eq: thm}
			\begin{align}
				& \underset{\mathcal Z_2}{{\normalfont \text{minimize}}} \qquad \!\!\! \Vert \pmb{\mathcal{S}^+} \mathcal{Z}_2 \Vert, \label{eq: con_thm_min}\\
				& \underset{}{{\normalfont \text{subject to}}} \qquad \!\!\! \pmb{\mathcal{M}} \mathcal{Z}_2 = \begin{bmatrix}
					\pmb{0}_{(n+m)\times(N-(n+m))}\\
					\I_{N-(n+m)}
				\end{bmatrix}\!.\label{eq: con_thm_con1}
			\end{align}
			Moreover, suppose there exist matrices $\Pi \in \R^{(n+m) \times (n+m)}$, where $\Pi \succ 0$, and $\mathcal{Y} \in \R^{\mathcal{T} \times (n+m)}$ such that
			\begin{align}
				& \pmb{\mathcal{M}} \mathcal{Y} = \begin{bmatrix}
					\Pi\\
					\pmb{0}_{(N-(n+m)) \times (n+m)}
				\end{bmatrix}\!, \label{eq: con_thm_con2}\\
				& \begin{bmatrix}
					(\frac{1}{1+\varpi}) \Pi & \pmb{\mathcal{S}^+} \mathcal{Y}\\
					\star & \Pi
				\end{bmatrix} \succeq 0, \label{eq: con_thm_con3}
			\end{align}
		\end{subequations}
		for some $\varpi \in \Rp$.
		Then, one can deduce that $\pmb{\mathds{B}_a}(\zeta) = \zeta^\top \mathcal{P} \zeta$, with $\mathcal{P} \coloneq \Pi^{-1}$, is an A-CBC for the unknown A-dt-GNS $\Upsilon_a$, with
		\begin{subequations}
			\begin{align}
				c_a & = \big(1 + \frac{1}{\varpi}\big) \lambda_{\max}(\mathcal{P}) \max_{\zeta \in \mathcal{X}^\zeta} \big\vert (\pmb{\mathcal{S}^+} \mathcal{Z}_2) \pmb{\Psi}(\zeta) \big\vert^2, \label{eq: con_thm_con4}\\
				\eta_a & = \lambda_{\max}(\mathcal{P}) \max_{\zeta \in \mathcal{X}^\zeta_{\mathrm{o}}} \vert \zeta \vert^2, \label{eq: con_thm_con5}\\
				\gamma_a & = \lambda_{\min}(\mathcal{P}) \min_{\zeta \in \mathcal{X}^\zeta_{\mathrm{1}}} \vert \zeta \vert^2, \label{eq: con_thm_con6}
			\end{align}
		\end{subequations}
		provided that $\gamma_a > \eta_a$. In addition,
		$\vartheta = \pmb{\mathcal{I}} \begin{bmatrix}
			\mathcal{Y} \mathcal{P} & \mathcal{Z}_2
		\end{bmatrix} \mathcal{F}(\zeta)$ is its associated safety controller.
	\end{thm}
	
	\begin{proof}
		Let us define $\mathcal{Z}_1 \coloneq \mathcal{Y}\mathcal{P}$. Since $\mathcal{P} = \Pi^{-1}$, it follows that fulfilling both conditions \eqref{eq: con_thm_con1} and \eqref{eq: con_thm_con2} together ensures the fulfillment of condition \eqref{eq: Lemma-con} in Lemma \reff{lemma 1}. Thus, under the control law \eqref{eq: Lemma-control}, we can conclude that we have the closed-loop  data-based representation \eqref{eq: Lemma-closed}.
		
		We now continue with demonstrating the satisfaction of condition \eqref{eq: A-con_bar_3} under the fulfillment of the SDP \eqref{eq: thm}. Since $\pmb{\mathds{B}_a}(\zeta) = \zeta^\top \mathcal{P} \zeta$, one has
		\begin{align*}
			&\pmb{\mathds{B}_a}(\mathcal{A} \mathcal{F}(\zeta) + \mathcal{B}\vartheta)\\
			&= (\mathcal{A} \mathcal{F}(\zeta) + \mathcal{B}\vartheta)^\top \mathcal{P} (\mathcal{A} \mathcal{F}(\zeta) + \mathcal{B}\vartheta)\\
			& \!\! \overset{\eqref{eq: Lemma-closed}}{=} ( \pmb{\mathcal{S}^+} \mathcal{Z}_1 \zeta + \pmb{\mathcal{S}^+} \mathcal{Z}_2 \pmb{\Psi}(\zeta))^\top \mathcal{P} ( \pmb{\mathcal{S}^+} \mathcal{Z}_1 \zeta + \pmb{\mathcal{S}^+} \mathcal{Z}_2 \pmb{\Psi}(\zeta))\\
			&= \zeta^\top [ \pmb{\mathcal{S}^+} \mathcal{Z}_1]^\top \mathcal{P} [ \pmb{\mathcal{S}^+} \mathcal{Z}_1] \zeta + \pmb{\Psi}^\top(\zeta) [ \pmb{\mathcal{S}^+} \mathcal{Z}_2]^\top \mathcal{P} [ \pmb{\mathcal{S}^+} \mathcal{Z}_2] \pmb{\Psi}(\zeta)\\
			& \hspace{0.3cm} + 2 \overbrace{\zeta^\top [ \pmb{\mathcal{S}^+} \mathcal{Z}_1]^\top  \sqrt{\mathcal{P}}}^{\text{\textcolor{black}{$a$}}} \overbrace{\sqrt{\mathcal{P}}  [ \pmb{\mathcal{S}^+} \mathcal{Z}_2] \pmb{\Psi}(\zeta)}^{\text{\textcolor{black}{$b$}}}.
		\end{align*}
		According to the Cauchy-Schwarz inequality \cite{bhatia1995cauchy}, \emph{i.e.,}  $a b \leq \vert a \vert \vert b \vert,$ for any $a^\top, b \in \R^{n+m}$, followed by
		employing Young's inequality \cite{young1912classes}, \emph{i.e.,} $\vert a \vert \vert b \vert \leq \frac{\varpi}{2} \vert a \vert^2 + \frac{1}{2\varpi} \vert b \vert^2$, for any $\varpi \in \Rp$, one has
		\begin{align*}
			& 2 \zeta^\top [ \pmb{\mathcal{S}^+} \mathcal{Z}_1]^\top  \sqrt{\mathcal{P}} \sqrt{\mathcal{P}}  [ \pmb{\mathcal{S}^+} \mathcal{Z}_2] \pmb{\Psi}(\zeta)\\
			& \leq \varpi \zeta^\top [ \pmb{\mathcal{S}^+} \mathcal{Z}_1]^\top \mathcal{P} [ \pmb{\mathcal{S}^+} \mathcal{Z}_1] \zeta + \frac{1}{\varpi}  \pmb{\Psi}^\top(\zeta) [ \pmb{\mathcal{S}^+} \mathcal{Z}_2]^\top \mathcal{P} [ \pmb{\mathcal{S}^+} \mathcal{Z}_2] \pmb{\Psi}(\zeta).
		\end{align*}
		Thus, considering $\mathcal{Z}_1 = \mathcal{Y} \mathcal{P}$, we have
		\begin{align}\notag
			\pmb{\mathds{B}_a}(\mathcal{A} \mathcal{F}(\zeta) \!+\! \mathcal{B}\vartheta) & \!\leq\! (1 \!+\! \varpi) \zeta^\top \!\mathcal{P} [ \pmb{\mathcal{S}^+} \mathcal{Y}]^\top \mathcal{P} [ \pmb{\mathcal{S}^+} \mathcal{Y}] \mathcal{P} \zeta\\\label{New54}
			& \hspace{0.3cm} \!+\! (1\!+\!\frac{1}{\varpi}) \pmb{\Psi}^\top\!(\zeta) [ \pmb{\mathcal{S}^+} \mathcal{Z}_2]^\top \!\mathcal{P} [ \pmb{\mathcal{S}^+} \mathcal{Z}_2] \pmb{\Psi}(\zeta).
		\end{align}
		Furthermore, according to the Schur complement \cite{zhang2006schur}, and considering  condition \eqref{eq: con_thm_con3} with $\mathcal{P} = \Pi^{-1}$, we know that
		\begin{align*}
			\begin{bmatrix}
				(\frac{1}{1+\varpi}) \mathcal{P}^{-1} & \pmb{\mathcal{S}^+} \mathcal{Y}\\
				\star & \mathcal{P}^{-1}
			\end{bmatrix} \succeq 0 & \Leftrightarrow \mathcal{P}^{-1} - (1+\varpi) [\pmb{\mathcal{S}^+} \mathcal{Y}]^\top \mathcal{P} [\pmb{\mathcal{S}^+} \mathcal{Y}] \succeq 0\\
			& \Leftrightarrow \mathcal{P} \!-\! (1\!+\!\varpi) \mathcal{P} [\pmb{\mathcal{S}^+} \mathcal{Y}]^\top \mathcal{P} [\pmb{\mathcal{S}^+} \mathcal{Y}] \mathcal{P} \succeq 0.
		\end{align*}
		Therefore, it is clear that
		\begin{align}
			(1 + \varpi) \zeta^\top \mathcal{P} [ \pmb{\mathcal{S}^+} \mathcal{Y}]^\top \mathcal{P} [ \pmb{\mathcal{S}^+} \mathcal{Y}] \mathcal{P} \zeta \leq \zeta^\top \mathcal{P} \zeta = \pmb{\mathds{B}_a}(\zeta). \label{eq:tmp1}
		\end{align}
       At the same time, we know from linear algebra that
		\begin{align*}
			\pmb{\Psi}^\top(\zeta) [ \pmb{\mathcal{S}^+} \mathcal{Z}_2]^\top \mathcal{P} [ \pmb{\mathcal{S}^+} \mathcal{Z}_2] \pmb{\Psi}(\zeta) \leq \lambda_{\max}(\mathcal{P}) \big\vert [ \pmb{\mathcal{S}^+} \mathcal{Z}_2] \pmb{\Psi}(\zeta) \big\vert^2.
		\end{align*}
		As one has
		\begin{align*}
			\lambda_{\max}(\mathcal{P}) \big\vert [ \pmb{\mathcal{S}^+} \mathcal{Z}_2] \pmb{\Psi}(\zeta) \big\vert^2 \leq \lambda_{\max}(\mathcal{P}) \max_{\zeta \in \mathcal{X}^\zeta} \big\vert (\pmb{\mathcal{S}^+} \mathcal{Z}_2) \pmb{\Psi}(\zeta) \big\vert^2,
		\end{align*}
		we can claim that
		\begin{align}
			\pmb{\Psi}^\top(\zeta) [ \pmb{\mathcal{S}^+} \mathcal{Z}_2]^\top \mathcal{P} [ \pmb{\mathcal{S}^+} \mathcal{Z}_2] \pmb{\Psi}(\zeta) \leq \lambda_{\max}(\mathcal{P}) \max_{\zeta \in \mathcal{X}^\zeta} \big\vert (\pmb{\mathcal{S}^+} \mathcal{Z}_2) \pmb{\Psi}(\zeta) \big\vert^2\!. \label{eq:tmp2}
		\end{align}
		Hence, one can conclude that the fulfillment of conditions \eqref{eq: thm} and \eqref{eq: con_thm_con4}, considering \eqref{eq:tmp1} and \eqref{eq:tmp2}, implies that
		\begin{align*}
			\pmb{\mathds{B}_a}(\mathcal{A} \mathcal{F}(\zeta) + \mathcal{B}\vartheta) \leq \pmb{\mathds{B}_a}(\zeta) + c_a,
		\end{align*}
		with $c_a = \big(1 + \frac{1}{\varpi}\big)\lambda_{\max}(\mathcal{P}) \max_{\zeta \in \mathcal{X}^\zeta} \big\vert (\pmb{\mathcal{S}^+} \mathcal{Z}_2) \pmb{\Psi}(\zeta) \big\vert^2,$ as introduced in \eqref{eq: con_thm_con4}, which demonstrates that condition \eqref{eq: A-con_bar_3} is also satisfied.
		
		To complete the proof, we now proceed to demonstrate the satisfaction of conditions \eqref{eq: A-con_bar_1} and \eqref{eq: A-con_bar_2}. Since it is well-known that
		\begin{align*}
			\pmb{\mathds{B}_a}(\zeta) = \zeta^\top \mathcal{P} \zeta \leq \lambda_{\max}(\mathcal{P}) \vert \zeta \vert^2,
		\end{align*}
		and for all $\zeta \in \mathcal{X}^\zeta_{\mathrm{o}},$ we know that
		\begin{align*}
			\lambda_{\max}(\mathcal{P}) \vert \zeta \vert^2 \leq \lambda_{\max}(\mathcal{P}) \max_{\zeta \in \mathcal{X}^\zeta_{\mathrm{o}}} \vert \zeta \vert^2,
		\end{align*}
		one can deduce that
		\begin{align*}
			\pmb{\mathds{B}_a}(\zeta) \leq \lambda_{\max}(\mathcal{P}) \max_{\zeta \in \mathcal{X}^\zeta_{\mathrm{o}}} \vert \zeta \vert^2, \qquad \forall \zeta \in \mathcal{X}^\zeta_{\mathrm{o}},
		\end{align*}
		depicting the satisfaction of \eqref{eq: A-con_bar_1} with $\eta_a = \lambda_{\max}(\mathcal{P}) \max_{\zeta \in \mathcal{X}^\zeta_{\mathrm{o}}} \vert \zeta \vert^2\!,$ as introduced in \eqref{eq: con_thm_con5}. Likewise, one can conclude that \eqref{eq: A-con_bar_2} is satisfied with $\gamma_a = \lambda_{\min}(\mathcal{P}) \min_{\zeta \in \mathcal{X}^\zeta_{\mathrm{1}}} \vert \zeta \vert^2\!,$ as introduced in \eqref{eq: con_thm_con6}, concluding the proof.
	\end{proof}
	
	\begin{remark}[\normalfont\textbf{Control over $c_a$}]\label{remark_min}
		It is preferable to have $c_a$ as small as possible, as this extends the safety guarantee over a longer horizon; thus, minimizing $\Vert \pmb{\mathcal{S}^+} \mathcal{Z}_2 \Vert$ can be highly beneficial. To achieve this, we enforce the decision variable $\mathcal{Z}_2$ to satisfy the equality condition \eqref{eq: con_thm_con1} while minimizing $\Vert \pmb{\mathcal{S}^+} \mathcal{Z}_2 \Vert$, as shown in condition \eqref{eq: con_thm_min}. Simultaneously, minimizing $\lambda_{\max}(\mathcal{P})$ would directly contribute to a smaller $c_a$, as it appears in \eqref{eq: con_thm_con4}. However, directly minimizing it is challenging, as it is an NP-hard problem. Instead, one can introduce the constraint $\Pi \succeq \kappa \I_{n+m}$, where $\kappa \in \Rpz$, and maximize $\kappa$ while solving \eqref{eq: con_thm_con2} and \eqref{eq: con_thm_con3}. In this approach, $\lambda_{\min}(\Pi)$ is maximized, and since $\Pi^{-1} = \mathcal{P}$, meaning $\lambda_{\max}(\Pi^{-1}) = \lambda_{\max}(\mathcal{P}) = \frac{1}{\lambda_{\min}(\Pi)}$, it follows that $\lambda_{\max}(\mathcal{P})$ is minimized, thereby minimizing $c_a$ as well.
	\end{remark}
	
	\begin{remark}[\normalfont\textbf{On $\eta_a$ and $\gamma_a$}]\label{remark_eta_gamma}
		Even though the values proposed for $\eta_a$ and $\gamma_a$ in \eqref{eq: con_thm_con5} and \eqref{eq: con_thm_con6} are always valid, they are endowed with inherent conservatism. To maximize the distance between these values, they can be optimized. Specifically, $\eta_a$ can be designed as the maximum value that $\pmb{\mathds{B}_a}(\zeta)$ attains in the initial set $\mathcal{X}^\zeta_{\mathrm{o}}$ (see \eqref{eq: A-con_bar_1}), while $\gamma_a$ can be determined as the minimum value of $\pmb{\mathds{B}_a}(\zeta)$ in the unsafe set $\mathcal{X}^\zeta_{\mathrm{1}}$ (see \eqref{eq: A-con_bar_2}). This approach results in a straightforward optimization problem that may provide a stronger guarantee.
	\end{remark}

    \begin{remark}[\normalfont\textbf{Quadratic A-CBCs}]
		We highlight that our framework simplifies the problem of safety controller design by finding an A-CBC for just the linear part of A-dt-GNS since the nonlinear part is pushed within $c_a$ (cf. \eqref{eq: con_thm_con4} and under Remark \reff{rem:tmp}), with its effect getting minimized within the compact domain (cf. \eqref{eq: con_thm_min} and \eqref{eq: con_thm_con1}). It is well-defined that for linear systems, the best choice of A-CBCs (similar to Lyapunov functions) is quadratic, facilitating problem-solving and scalability.
	\end{remark}
	
	Having offered the paper's main theorem, we now present the required steps to design an A-CBC and its safety controller from data. Specifically, from \eqref{eq: con_thm_min} and \eqref{eq: con_thm_con1}, one can design $\mathcal{Z}_2$, while \eqref{eq: con_thm_con2} and \eqref{eq: con_thm_con3} yield $\Pi$ (and, subsequently, $\mathcal{P}$) as well as $\mathcal{Y}$. With $\mathcal{Z}_2$, $\mathcal{P}$, and $\mathcal{Y}$ determined, $c_a$ can be straightforwardly computed from \eqref{eq: con_thm_con4} (\emph{e.g.,} by utilizing \texttt{fmincon} in \textsc{Matlab}), and $\eta_a$ and $\gamma_a$ can be computed using Remark \reff{remark_eta_gamma}.
	
	A potential difficulty with this approach, however, is that it may result in a relatively large value for $c_a$, thus decreasing the safety time horizon $T \in \Np$ that must satisfy $T < \frac{\gamma_a - \eta_a}{c_a}$. This raises a valid question about exploring an alternative computational approach, where one could directly design $\mathcal{P}$ and $c_a$ by bounding the nonlinear part in $\pmb{\mathds{B}_a}(\mathcal{A} \mathcal{F}(\zeta) + \mathcal{B}\vartheta)$ (\emph{i.e.,} $\pmb{\Psi}^\top(\zeta) [ \pmb{\mathcal{S}^+} \mathcal{Z}_2]^\top \mathcal{P} [ \pmb{\mathcal{S}^+} \mathcal{Z}_2]\\ \pmb{\Psi}(\zeta)$). With the obtained $\mathcal{P}$, one can then design $\mathcal{Y}$ from \eqref{eq: con_thm_con2} and \eqref{eq: con_thm_con3}, while ensuring the minimum possible value for $c_a$. In this way, not only can an A-CBC be found (if one exists), but it also ensures the best possible guarantee. The subsequent section outlines this approach in detail.
	
	\section{Alternative computation for A-CBCs and safety controllers} \label{Section: bound}
	In this section, we focus on designing the matrix $\mathcal{P}$ and the constant $c_a$, followed by solving \eqref{eq: con_thm_con2} and \eqref{eq: con_thm_con3} to obtain $\mathcal{Y}$, the A-CBC, and its associated safety controller.
	Within this alternative computational approach, we assume that $\mathcal{Z}_2$ has already been obtained from \eqref{eq: con_thm_min} and \eqref{eq: con_thm_con1}. Therefore, the problem we aim to solve can be formulated as the following robust convex program (RCP):
	\begin{mini!}|l|[2]
		{\mathcal{P}, c_a, \mu}{\mu + c_a,}{\label{ROP}}\notag
		{}
		\addConstraint{\forall \zeta \in X^\zeta, \mathcal{P}\succ 0, c_a \in \Rpz, \mu \in (-\infty, 0], \varpi \in \Rp}{}
		\addConstraint{\big(1 + \frac{1}{\varpi}\big)\pmb{\Psi}^\top(\zeta) [ \pmb{\mathcal{S}^+} \mathcal{Z}_2]^\top \mathcal{P} [ \pmb{\mathcal{S}^+} \mathcal{Z}_2] \pmb{\Psi}(\zeta) - c_a}{\leq \mu,}\label{ROP1}
	\end{mini!}
   where $\varpi$ should be fixed a-priori, similar to \eqref{eq: con_thm_con3}.
	The optimal decision variables of the RCP \eqref{ROP} are shown by $\mu_R^*, c_{a_R}^*,$ and $\mathcal{P}_R^*$. It is clear that any feasible solution to the RCP \eqref{ROP} confirms the satisfaction of
	\begin{align}\label{New87}
		\text{$\big(1 + \frac{1}{\varpi}\big)$}\pmb{\Psi}^\top(\zeta) [ \pmb{\mathcal{S}^+} \mathcal{Z}_2]^\top \mathcal{P} [ \pmb{\mathcal{S}^+} \mathcal{Z}_2] \pmb{\Psi}(\zeta) \!\leq \! c_a,
	\end{align}
    with $\mathcal{P} \!\succ\! 0,$ and $c_a\! \in \! \Rpz,$
	while potentially minimizing $c_a$, which is the main goal of this alternative approach. In other words, the left-hand side in \eqref{New87} corresponds exactly to the last term in \eqref{New54}, which leads to $c_a$ in \eqref{eq: con_thm_con4} within the computational approach discussed in the previous section. Here, the primary goal is to bound this term by minimizing $c_a$ through the optimization problem in \eqref{ROP}, potentially leading to extended safety guarantees over time.
	
	The RCP \eqref{ROP}, however, presents a significant challenge, making it intractable. Specifically, due to the continuous nature of $\mathcal{X}^\zeta$ and the presence of $\pmb{\Psi}(\zeta)$ in \eqref{ROP1}, there are infinitely many constraints. To address this issue, we evaluate $\pmb{\Psi}(\zeta)$ at sample points $\zeta^d$, where $d \in \{1, 2, \dots, \mathcal{N}\}$, and $\mathcal{N}$ denotes the number of sample points for evaluation. Subsequently, we compute the maximum distance between $\zeta \in \mathcal{X}^\zeta$, and the assessment samples as
	\begin{align}
		\delta = \max_{\zeta} \min_{d} \vert \zeta - \zeta^d \vert, \quad \forall \zeta \in \mathcal{X}^\zeta. \label{eq: disc_param}
	\end{align}
	It should be noted that the maximum distance $\delta$ can be computed via grid-based partitioning of the space $\mathcal{X}^\zeta$.
	
	\begin{remark}[\normalfont\textbf{Dynamics Exclusion in Alternative Analysis}]
		We stress the fact that assessing $\pmb{\Psi}(\zeta)$ at sample points $\zeta^d$ does not require collecting data points from the system, meaning the system does not need to be actively running to obtain these samples. More precisely, since our analysis here is focused solely on the known state set $\mathcal{X}^\zeta$, we can select sample points $\zeta^d$ within this set. This approach is sufficient for proceeding with our alternative computational method.
	\end{remark}
	
	Leveraging the assessment samples $\zeta^d \in \mathcal{X}^\zeta, d \in \{1, 2, \dots, \mathcal{N}\}$, we introduce the following scenario convex program (SCP) that corresponds to the RCP \eqref{ROP}:
	\begin{mini!}|l|[2]
		{\mathcal{P}, c_a, \mu}{\mu + c_a,}{\label{SCP}}\notag
		{}
		\addConstraint{ \forall \zeta^d \in  \mathcal{X}^\zeta, \mathcal{P}\succ 0, c_a \in \Rpz, \mu \in (-\infty, 0], \varpi \in \Rp}{}
		\addConstraint{\forall d \in  \{1, 2, \dots, \mathcal{N}\}}{}\notag
		\addConstraint{\big(1 + \frac{1}{\varpi}\big)\pmb{\Psi}^\top(\zeta^d) [ \pmb{\mathcal{S}^+} \mathcal{Z}_2]^\top \mathcal{P} [ \pmb{\mathcal{S}^+} \mathcal{Z}_2] \pmb{\Psi}(\zeta^d) - c_a}{\leq \mu,}\label{SCP1}
	\end{mini!}
	where $\varpi$ should be selected a-priori. The optimal decision variables of the SCP \eqref{SCP} are shown by $\mu_S^*, c_{a_S}^*,$ and $\mathcal{P}_S^*$.
	One can readily deduce that the proposed SCP \eqref{SCP} contains a finite number of constraints, all of which have the same form as those in the RCP \eqref{ROP}.
	
	In the following subsections, we offer two different approaches to establish a relation between a feasible solution of the SCP \eqref{SCP} and that of the RCP \eqref{ROP}, while providing correctness guarantees.
	
	\subsection{Deterministic correctness guarantee}\label{subsec_det}
	Given the fact that the dictionary \eqref{eq: final_dict} contains smooth nonlinear functions, and $\mathcal{X}^\zeta$ is a compact domain, one can conclude that
	$\text{$\big(1 + \frac{1}{\varpi}\big)$} \pmb{\Psi}^\top(\zeta) [ \pmb{\mathcal{S}^+} \mathcal{Z}_2]^\top \mathcal{P} [ \pmb{\mathcal{S}^+} \mathcal{Z}_2] \pmb{\Psi}(\zeta)$ is always Lipschitz continuous with respect to $\zeta$, with a Lipschitz constant denoted as $\mathscr{L}$. In the following proposition, we demonstrate that, under a specific condition involving the Lipschitz constant $\mathscr{L}$, the maximum distance $\delta$ \eqref{eq: disc_param}, and the optimal solution $\mu^*_S$, a solution for the SCP is also valid for the RCP with a correctness guarantee.
	
	\begin{prop}[\normalfont\textbf{Deterministic Correctness Guarantee}] \label{thm_deterministic}
		Consider the SCP \eqref{SCP} with its solutions $\mathcal{P}_S^*, c_{a_S}^*,$ and $\mu_S^*$ obtained via $\mathcal{N}$ assessment samples. If
		\begin{align}
			\mu_S^* + \mathscr{L}\delta \leq 0, \label{eq: thm-conf1}
		\end{align}
		then the solutions of the SCP \eqref{SCP} are valid for the RCP \eqref{ROP} with a correctness guarantee.
	\end{prop}
	
	\begin{proof}
		Let us define $\hat d \coloneq \arg\: \underset{d}{\min}$ $\vert \zeta^d - \zeta \vert$.
		Given the left-hand side of \eqref{ROP1} as
		\begin{align*}
			\text{$\big(1 + \frac{1}{\varpi}\big)$}\pmb{\Psi}^\top(\zeta) [ \pmb{\mathcal{S}^+} \mathcal{Z}_2]^\top \mathcal{P} [ \pmb{\mathcal{S}^+} \mathcal{Z}_2] \pmb{\Psi}(\zeta) - c_a,
		\end{align*}
		and by incorporating the term $$\text{$\big(1 + \frac{1}{\varpi}\big)$}\pmb{\Psi}^\top(\zeta^{\hat d}) [ \pmb{\mathcal{S}^+} \mathcal{Z}_2]^\top \mathcal{P} [ \pmb{\mathcal{S}^+} \mathcal{Z}_2] \text{$\pmb{\Psi}(\zeta^{\hat d})$} - c_a$$ through addition and subtraction, we have 
		\begin{align}\notag
			&\text{$\big(1 + \frac{1}{\varpi}\big)$} \pmb{\Psi}^{\!\top}\!(\zeta) [ \pmb{\mathcal{S}^+} \mathcal{Z}_2]^\top \mathcal{P} [ \pmb{\mathcal{S}^+} \mathcal{Z}_2] \pmb{\Psi}(\zeta) - c_a \\\notag
			& \hspace{0.3cm}+ \text{$\big(1 + \frac{1}{\varpi}\big)$} \pmb{\Psi}^\top(\zeta^{\hat d}) [ \pmb{\mathcal{S}^+} \mathcal{Z}_2]^\top \mathcal{P} [ \pmb{\mathcal{S}^+} \mathcal{Z}_2] \pmb{\Psi}(\zeta^{\hat d}) - c_a\\\notag
			& \hspace{0.3cm} - \text{$\big(1 + \frac{1}{\varpi}\big)$} \pmb{\Psi}^\top(\zeta^{\hat d}) [ \pmb{\mathcal{S}^+} \mathcal{Z}_2]^\top \mathcal{P} [ \pmb{\mathcal{S}^+} \mathcal{Z}_2] \pmb{\Psi}(\zeta^{\hat d}) + c_a\\\notag
			& = \text{$\big(1 + \frac{1}{\varpi}\big)$} \pmb{\Psi}^{\!\top}\!(\zeta) [ \pmb{\mathcal{S}^+} \mathcal{Z}_2]^\top \mathcal{P} [ \pmb{\mathcal{S}^+} \mathcal{Z}_2] \pmb{\Psi}(\zeta) \\\notag
			& \hspace{0.3cm}- \text{$\big(1 + \frac{1}{\varpi}\big)$} \pmb{\Psi}^\top(\zeta^{\hat d}) [ \pmb{\mathcal{S}^+} \mathcal{Z}_2]^\top \mathcal{P}  [ \pmb{\mathcal{S}^+} \mathcal{Z}_2] \pmb{\Psi}(\zeta^{\hat d}) \\\label{new87}
			& \hspace{0.3cm}+ \text{$\big(1 + \frac{1}{\varpi}\big)$} \pmb{\Psi}^\top(\zeta^{\hat d}) [ \pmb{\mathcal{S}^+} \mathcal{Z}_2]^\top \mathcal{P} [ \pmb{\mathcal{S}^+} \mathcal{Z}_2] \pmb{\Psi}(\zeta^{\hat d}) - c_a.
		\end{align}
		Since the first two terms in \eqref{new87} can be bounded using the Lipschitz constant $\mathscr{L}$, and $\mu_S^*$ can bound the last two terms according to \eqref{SCP1}, one has
		\begin{align*}
			& \text{$\big(1 + \frac{1}{\varpi}\big)$}\pmb{\Psi}^\top(\zeta) [ \pmb{\mathcal{S}^+} \mathcal{Z}_2]^\top \mathcal{P} [ \pmb{\mathcal{S}^+} \mathcal{Z}_2] \pmb{\Psi}(\zeta) - c_a \\
			&\leq \mathscr{L} \vert \zeta^{\hat d} - \zeta \vert + \mu_S^* \leq \mathscr{L} \min_{d} \vert \zeta^d - \zeta \vert + \mu_S^*\\
			& \leq \mathscr{L} \max_{\zeta} \: \min_{d} \vert \zeta^d - \zeta \vert + \mu_S^* \overset{\eqref{eq: disc_param}}{=} \mathscr{L} \delta + \mu_S^*\overset{\eqref{eq: thm-conf1}}{\leq} 0.
		\end{align*}
		Thus, one can deduce that condition \eqref{ROP1} is satisfied with $\mu_R^* = \mathscr{L} \delta + \mu_S^* \leq 0$, concluding the proof.
	\end{proof}
	
	We now provide some remarks on the results of Proposition \reff{thm_deterministic}. First, it should be noted that the proposed condition \eqref{eq: thm-conf1} must be verified a posteriori, meaning that the SCP \eqref{SCP} should be solved first, followed by a check to determine whether condition \eqref{eq: thm-conf1} is satisfied. If this condition is not met, the SCP \eqref{SCP} must be solved again with an increased number of assessment samples, thus reducing $\delta$ and potentially aiding in satisfying condition \eqref{eq: thm-conf1}. Second, since our aim is to satisfy condition \eqref{eq: thm-conf1}, we define the cost function as $\mu + c_a$, which does not necessarily yield the smallest possible $c_a$ but rather the smallest value of $c_a$ that minimizes $\mu + c_a$. Finally, it is worth noting that, since $\pmb{\Psi}(\zeta)$ is available and $\mathcal{Z}_2$ has already been determined, finding the Lipschitz constant $\mathscr{L}$ becomes a straightforward optimization problem in a compact domain once $\mathcal{P}$ is computed.
	
	In scenarios where condition~\eqref{eq: thm-conf1} is not satisfied, the results of the SCP are not valid for the RCP. In the following subsection, we present an alternative approach for minimizing $c_a$ that does not require this condition but at the cost of offering a probabilistic correctness guarantee.
	
	\subsection{Probabilistic correctness guarantee}\label{Sub3}
	Here, to obtain a potentially small $c_a$ without needing to satisfy condition~\eqref{eq: thm-conf1}, we adjust the RCP \eqref{ROP}, and its corresponding SCP \eqref{SCP} accordingly, by removing $\mu$ from the cost function to focus solely on $c_a$. The adjusted RCP is formulated as
	\begin{mini!}|l|[2]
		{\mathcal{P}, c_a}{c_a,}{\label{SCP_ad}}\notag
		{}
		\addConstraint{ \forall \zeta \in  \mathcal{X}^\zeta, \mathcal{P}\succ 0, c_a \in \Rpz, \varpi \in \Rp}{}
		\addConstraint{\big(1 + \frac{1}{\varpi}\big) \pmb{\Psi}^\top(\zeta) [ \pmb{\mathcal{S}^+} \mathcal{Z}_2]^\top \mathcal{P} [ \pmb{\mathcal{S}^+} \mathcal{Z}_2] \pmb{\Psi}(\zeta)}{\leq c_a,}\label{SCP1_ad}
	\end{mini!}
	where $\varpi$ should be chosen a-priori. This is then followed by defining two key parameters: (i) the violation parameter $\varepsilon \in (0, 1)$, and (ii) the confidence parameter $\beta \in (0, 1)$. The violation parameter $\varepsilon$ can also be interpreted as a \emph{risk probability}, indicating that one allows
	\begin{align*}
		\text{$\big(1 + \frac{1}{\varpi}\big)$} \pmb{\Psi}^\top(\zeta) [ \pmb{\mathcal{S}^+} \mathcal{Z}_2]^\top \mathcal{P} [ \pmb{\mathcal{S}^+} \mathcal{Z}_2] \pmb{\Psi}(\zeta) \!\leq \! c_a,
	\end{align*}
    with $\mathcal{P} \!\succ\! 0,$ and $c_a\! \in \! \Rpz,$
	to hold throughout the state set $\mathcal{X}^\zeta$, except for a fraction $\varepsilon$. This means that $\varepsilon$ is an upper bound on the probability that
	\begin{align*}
		\text{$\big(1 + \frac{1}{\varpi}\big)$}\pmb{\Psi}^\top(\zeta) [ \pmb{\mathcal{S}^+} \mathcal{Z}_2]^\top \mathcal{P} [ \pmb{\mathcal{S}^+} \mathcal{Z}_2] \pmb{\Psi}(\zeta) \nleq  c_a,
	\end{align*}
	or equivalently,
	\begin{align*}
		\mathrm{Pr} \big[\text{$\big(1 + \frac{1}{\varpi}\big)$}\pmb{\Psi}^\top(\zeta) [ \pmb{\mathcal{S}^+} \mathcal{Z}_2]^\top \mathcal{P} [ \pmb{\mathcal{S}^+} \mathcal{Z}_2] \pmb{\Psi}(\zeta) \leq  c_a\big] \geq 1 - \varepsilon,
	\end{align*}
	where $\mathrm{Pr}(\cdot)$ denotes the probability of an event.
	The second parameter, $\beta$, represents the confidence that the risk will not exceed the specified level $\varepsilon$. Once these parameters are set, the number of i.i.d. assessment samples $\mathcal{N}$, needed to design $\mathcal{P}$ and $c_a$, can be determined based on the following theorem, borrowed from \cite{campi2009scenario}.
	
	\begin{thm}[\normalfont\textbf{Probabilistic Correctness Guarantee}]\label{thm: prob}
		Consider a violation parameter $\varepsilon \in (0, 1)$ and a confidence parameter $\beta \in (0, 1)$. If one chooses
		\begin{align}
			\mathcal{N} \geq \frac{2}{\varepsilon} \left(\ln \left(\frac{1}{\beta}\right) + \frac{(n+m)(n+m+1)}{2} + 1\right)\! \label{eq: number}
		\end{align}
		and solves the adjusted RCP \eqref{SCP_ad} using these samples, then it holds over $\mathcal{X}^\zeta$ with the following probability:
		\begin{align}
			\mathrm{Pr}&\Big[\mathrm{Pr} \big[\text{$\big(1 + \frac{1}{\varpi}\big)$}\pmb{\Psi}^\top(\zeta) [ \pmb{\mathcal{S}^+} \mathcal{Z}_2]^\top \mathcal{P} [ \pmb{\mathcal{S}^+} \mathcal{Z}_2] \pmb{\Psi}(\zeta) \leq  c_a\big] \geq 1 - \varepsilon\Big] \notag \\ & 
			\text{$\geq 1 - \beta.$} \label{eq: prob_s}
		\end{align}
	\end{thm}
	
	\begin{remark}[\normalfont\textbf{On condition \eqref{eq: number}}]
		Note that in \eqref{eq: number}, $ \frac{(n+m)(n+m+1)}{2} + 1$ represents the number of decision variables in \eqref{SCP_ad}. Specifically, since $P \in \R^{(n+m) \times (n+m)}$ is a symmetric matrix, it contributes $\frac{(n+m)(n+m+1)}{2}$ decision variables, while $c_a \in \Rpz$ adds one additional decision variable. Hence, $ \frac{(n+m)(n+m+1)}{2} + 1$ is the total of decision variables in RCP \eqref{SCP_ad}.
	\end{remark}
	
	Having presented Theorem \reff{thm: prob}, which provides a probabilistic correctness guarantee, we now proceed with the following result, extending this probability to the A-CBC and its associated safety controller.
	
	\begin{coro}[\normalfont\textbf{Probabilistic A-CBCs and safety controllers}]
		Consider the unknown A-dt-GNS $\Upsilon_a$ as per \eqref{eq: final_sys} with its closed-loop data-based representation $\mathcal{A} \mathcal{F}(\zeta) + \mathcal{B}\vartheta = \pmb{ \mathcal{S}^+} \mathcal{Z}_1 \zeta + \pmb{\mathcal{S}^+} \mathcal{Z}_2 \pmb{\Psi}(\zeta)$ as per Lemma \reff{lemma 1}. Let $\mathcal{Z}_2$ be obtained from \eqref{eq: con_thm_min},\eqref{eq: con_thm_con1}, and $\mathcal{P}$ and $c_a$ are determined according to Theorem \reff{thm: prob}. Subsequently, if $\mathcal{Y}$ and $\Pi = \mathcal{P}^{-1}$ can be designed from \eqref{eq: con_thm_con2},\eqref{eq: con_thm_con3}, then $\pmb{\mathds{B}_a}(\zeta) = \zeta^\top \mathcal{P} \zeta$ is an A-CBC for the A-dt-GNS and $\vartheta = \pmb{\mathcal{I}} \begin{bmatrix}
			\mathcal{Y} \mathcal{P} & \mathcal{Z}_2
		\end{bmatrix} \mathcal{F}(\zeta)$ is its corresponding safety controller with a confidence of at least $1-\beta$.
	\end{coro}
	
	\begin{proof}
		Let $\mathcal{E}_1$ be the event that $\mathcal{Z}_2$, $\mathcal{Y},$ and $\Pi = \mathcal{P}^{-1}$ can be designed from the satisfaction of  \eqref{eq: con_thm_min}-\eqref{eq: con_thm_con3}. Additionally, let $\mathcal{E}_2$ represent the event \begin{align*}
			\mathrm{Pr} \big[\text{$\big(1 + \frac{1}{\varpi}\big)$}\pmb{\Psi}^\top(\zeta) [ \pmb{\mathcal{S}^+} \mathcal{Z}_2]^\top \mathcal{P} [ \pmb{\mathcal{S}^+} \mathcal{Z}_2] \pmb{\Psi}(\zeta) \leq  c_a\big] \geq 1 - \varepsilon,
		\end{align*} \emph{i.e., } the inner probability expression in \eqref{eq: prob_s}. We have $\mathrm{Pr}\big[\mathcal{\bar {\mathcal{E}}}_1\big]=0$ as $\mathcal{E}_1$ is a deterministic event and always holds true, and  $\mathrm{Pr}\big[\mathcal{\bar {\mathcal{E}}}_2\big]\leq \beta,$ where $\mathcal{\bar {\mathcal{E}}}_1$ and $\mathcal{\bar {\mathcal{E}}}_2$ are the complement of $\mathcal{E}_1$ and $\mathcal{E}_2$, respectively. Our goal is to compute the probability of the simultaneous occurrence of events  $\mathcal E_1$ and $\mathcal{E}_2$ as
		\begin{align*}
			\mathrm{Pr}\big[\mathcal E_1\cap \mathcal{E}_2\big]=1-\mathrm{Pr}\big[\mathcal{\bar {\mathcal{E}}}_1\cup \mathcal{\bar {\mathcal{E}}}_2\big].
		\end{align*}
		Since $\mathrm{Pr}\big[\mathcal{\bar {\mathcal{E}}}_1\cup \mathcal{\bar {\mathcal{E}}}_2\big]\leq\mathrm{Pr}\big[\mathcal{\bar {\mathcal{E}}}_1\big]+\mathrm{Pr}\big[\mathcal{\bar {\mathcal{E}}}_2\big]$,
		we have 
		\begin{align}\label{Eq:18}
			&\mathrm{Pr}\big[\mathcal E_1\cap \mathcal{E}_2\big]\geq 1-\mathrm{Pr}\big[\mathcal{\bar {\mathcal{E}}}_1\big]-\mathrm{Pr}\big[\mathcal{\bar {\mathcal{E}}}_2\big]
			\geq 1-\beta.
		\end{align}
		Hence, all required conditions for the design of $\pmb{\mathds{B}_a}(\zeta) = \zeta^\top \mathcal{P} \zeta$ as an A-CBC for the A-dt-GNS and $\vartheta = \pmb{\mathcal{I}} \begin{bmatrix} \mathcal{Y} \mathcal{P} & \mathcal{Z}_2 \end{bmatrix} \mathcal{F}(\zeta)$ as its corresponding safety controller are satisfied with a confidence level of at least $1 - \beta$, thereby concluding the proof.
	\end{proof}
	
	\section{Simulation Results}\label{Section: Simul}
    In this section, we shed light on the efficacy of our data-driven framework through two case studies: (i) a highly nonlinear system and (ii) a relatively high-dimensional system, highlighting the scalability of our approach.
    
	\textbf{First Case Study.} We demonstrate the effectiveness of our data-driven methodology in ensuring the safety of a highly nonlinear example while enforcing input constraints. The dt-GNS evolves according to 
	\begin{align}
		\Upsilon : \begin{cases}
			x^+_1 = x_1 + \tau\big(x_2 + \sin(u)\big),\\
			x^+_2 = x_2 + \tau\big(-x_1 + x_2 - \ln\big(1+x_1^2\big) + u\big),
		\end{cases} \label{eq: ex1_sys}
	\end{align}
	where $\tau = 0.01$ is the sampling time. Moreover, $\mathcal{C}_1 = -\frac{1}{15}$ and $\mathcal{C}_2 = \frac{1}{15}$ in \eqref{eq: input_const}, which implies that the control input is constrained by $-15 \leq u \leq 15$. We assume that the dt-GNS \eqref{eq: ex1_sys} is unknown, and just the dictionary
	\begin{align}
		f(x, u) = &\left[\begin{array}{cccccc}
			\!\!\! x_1 & x_2 & u & \ln(1+x_1^2) & \ln(1+x_2^2) & \ln(1+u^2)
		\end{array}\right.\notag\\
		&\left. \begin{array}{cccc}
			\!\!\cos(x_1) & \cos(x_2) & \cos(u) & \sin(u)
		\end{array}\!\!\!\right]^{\! \top} \label{dic1}
	\end{align}
\emph{with irrelevant terms} is available based on the insight into the system. The dt-GNS in \eqref{eq: ex1_sys} can be reformulated in the form of \eqref{eq: gen_disc_rewritten} using the dictionary \eqref{dic1} and the matrix $A$ as
\begin{align}
	A = \begin{bmatrix}
		1 & \tau & 0 & 0 & 0 & 0 & 0 & 0 & 0 & \tau\\
		-\tau & 1+\tau & \tau & -\tau & 0 & 0 & 0 & 0 & 0 & 0 
	\end{bmatrix}\!, \label{app1}
\end{align}
however, this matrix is unknown to us. Note that if the dictionary \eqref{dic1} changes, this unknown matrix $A$ will also change correspondingly to match the dt-GNS in \eqref{eq: ex1_sys}.

The regions of interest are given as $\mathcal{X} = [-5, 5]^2, \mathcal{X}_{\mathrm{o}} = [-1, 1]^2,$ and $\mathcal{X}_1 = [-5, -3]^2 \cup [3, 5]^2$. Now, according to the procedure introduced in Subsection \reff{subsec:a}, we construct the A-dt-GNS $\Upsilon_a$. Defining $\zeta_{1,1} \coloneq x_1, \zeta_{1, 2} \coloneq x_2,$ and $\zeta_2 \coloneq u$, and considering $\vartheta$ as the virtual control input, the A-dt-GNS is obtained as
	\begin{align}
		\Upsilon_a : \begin{cases}
			\zeta_{1,1}^+ = \zeta_{1,1} + \tau\big(\zeta_{1,2} + \sin(\zeta_2)\big),\\
			\zeta_{1,2}^+ = \zeta_{1,2} + \tau\big(-\zeta_{1,1} + \zeta_{1,2} - \ln\big(1+\zeta_{1,1}^2\big) + \zeta_2\big),\\
			\zeta_2^+ = \vartheta,
		\end{cases} \label{eq: ex1_sys_aug}
	\end{align}
	with the dictionary
	\begin{align}
		\mathcal{F}(\zeta) = &\left[\begin{array}{cccccc}
			\!\!\! \zeta_{1,1} & \zeta_{1,2} & \zeta_2 & \ln(1+\zeta_{1,1}^2) & \ln(1+\zeta_{1,2}^2) & \ln(1+\zeta_2^2)
		\end{array}\right.\notag\\
		&\left. \begin{array}{cccc}
			\!\!\cos(\zeta_{1,1}) & \cos(\zeta_{1,2}) & \cos(\zeta_2) & \sin(\zeta_2)
		\end{array}\!\!\!\right]^{\! \top}\!\!\!\!, \label{dic2}
	\end{align}
	where $\zeta = \left[\begin{array}{ccc} \!\!\! \zeta_{1,1} & \zeta_{1,2} & \zeta_2\end{array}\!\!\!\right]^{\! \top}$. The augmented system in \eqref{eq: ex1_sys_aug} can be recast in the form of \eqref{eq: final_sys}, with matrices $\mathcal{A}$ and $\mathcal{B}$ as
	\begin{align}
		\mathcal{A} \!=\! \begin{bmatrix}
			1 & \tau & 0 & 0 & 0 & 0 & 0 & 0 & 0 & \tau\\
			-\tau & 1+\tau & \tau & -\tau & 0 & 0 & 0 & 0 & 0 & 0 \\
			0 & 0 & 0 & 0 & 0 & 0 & 0 & 0 & 0 & 0
		\end{bmatrix}\!,~
		\mathcal{B} \!=\! \begin{bmatrix}
			0\\
			0\\
			1
		\end{bmatrix}\!. \label{app3}
	\end{align}
    We stress the fact that both the dt-GNS \eqref{eq: ex1_sys} and A-dt-GNS \eqref{eq: ex1_sys_aug} are assumed to be unknown (\emph{i.e.,} matrices $A$ and $\mathcal A$), and we just have the dictionary \eqref{dic1} which results in constructing \eqref{dic2}. We now proceed with building the sets of the A-dt-GNS \eqref{eq: ex1_sys_aug}. Considering $\epsilon_1 = \frac{1}{1500}$ and $\epsilon_2 = \frac{1499}{1500}$ as arbitrary parameters (cf. Remark \reff{remark 2}), the sets of A-dt-GNS are constructed as $\mathcal{X}^\zeta \coloneq [-5, 5]^2 \times [-15.01, 15.01],$ $\mathcal{X}^\zeta_{\mathrm{o}} \coloneq [-1, 1]^2 \times [-0.01, 0.01],$ and $\mathcal{X}^\zeta_{\mathrm{1}} \coloneq [-5, -3]^2 \times [-15.01, 15.01] \cup [3, 5]^2 \times [-15.01, 15.01] \cup [-5, 5]^2 \times [-15.01, -15] \cup [-5, 5]^2 \times [15, 15.01]$. Now, the primary objective is to design an A-CBC and its corresponding safety controller for the A-dt-GNS \eqref{eq: ex1_sys_aug} to guarantee its safety and, consequently, according to Remark \reff{Re_New}, to guarantee the safety of dt-GNS \eqref{eq: ex1_sys}, while enforcing its input constraint.
	
		\begin{figure}[t!]
		\centering
		\includegraphics[width=0.8\linewidth]{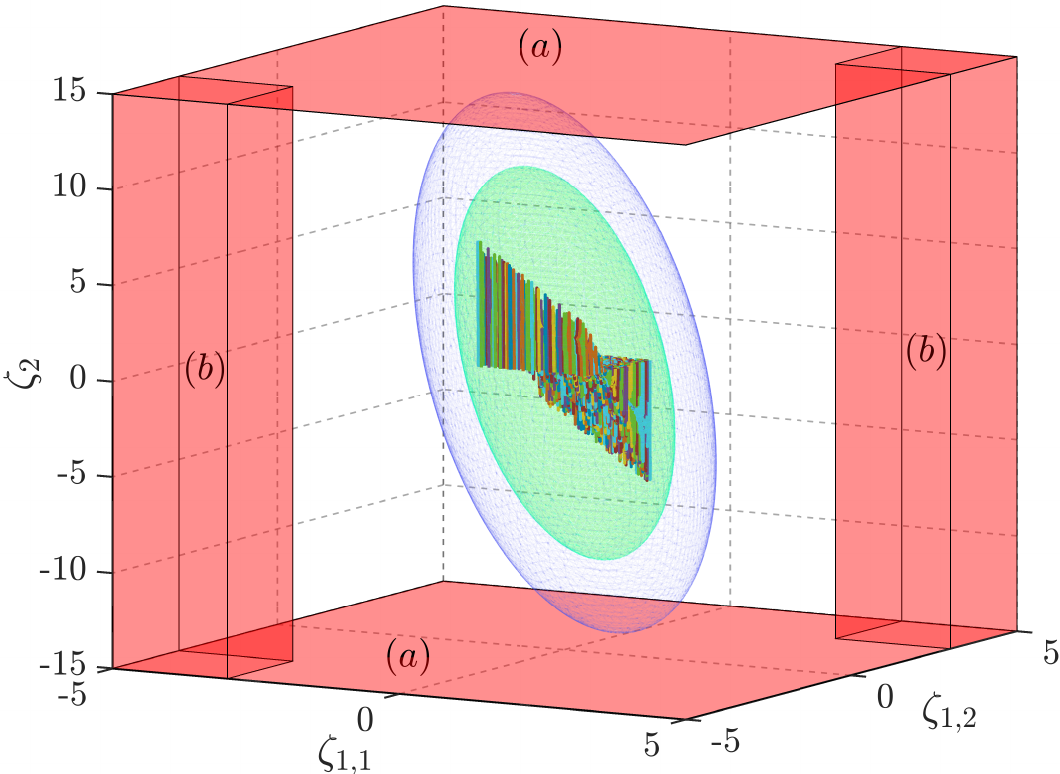}
		\caption{Starting from various initial conditions in the initial set $\mathcal{X}^\zeta_{\mathrm{o}}$, $1000$ state and control input trajectories of the unknown dt-GNS \eqref{eq: ex1_sys} are generated using the designed dynamic safety controller \eqref{des_u}. As illustrated, neither the states enter the unsafe set nor is the control input constraint violated, \emph{i.e.,} neither of them enters the unsafe sets \protect \redsquare. The unsafe sets marked by $(a)$ are to ensure the input constraints are met, while the unsafe sets labeled by $(b)$ are to ensure the safety of the system's states. Moreover, $\pmb{\mathds{B}_a} (\zeta) = \eta_a$ and $\pmb{\mathds{B}_a} (\zeta) = \gamma_a$ are indicated by colors \protect \greensquare ~ and \protect \bluesquare, respectively. }
		\label{fig:states}
	\end{figure}
	
	Having completed these steps, we now collect input-state data as in \eqref{eq: data} over the time horizon $\mathcal{T} = 11$ ($\mathcal{T} \geq N + 1$ with $N = 10$). We also set $\varpi = 0.01.$ We first examine the initial computational approach proposed in Section \reff{Section: Data-driven-disc}. After solving SDPs \eqref{eq: thm}, we obtain
	\begin{align*}
		\mathcal{P} = \begin{bmatrix}
			20.1130 & 12.1850 & 1.2337\\ 12.1850 & 22.0934 & 2.7004\\ 1.2337 & 2.7004 & 0.8938
		\end{bmatrix}\!.
	\end{align*}
	Therefore, the A-CBC is designed as
	\begin{align*}
		\pmb{\mathds{B}_a}(\zeta) = & \; 20.1130 \zeta_{1,1}^2 + 22.0934 \zeta_{1,2}^2 + 0.8938 \zeta_2^2 + 24.3700 \zeta_{1,1} \zeta_{1,2}\\
		& + 2.4674 \zeta_{1,1} \zeta_2 + 5.4008 \zeta_{1,2} \zeta_2,
	\end{align*}
	with $c_a = 3.9385$ being obtained from \eqref{eq: con_thm_con4}, and $\eta_a = 66.6553, \gamma_a = 125.7459,$ where $\gamma_a > \eta_a$, being acquired according to Remark \reff{remark_eta_gamma}. Hence, the virtual static safety controller $\vartheta$ is designed as
	\begin{align*}
		\vartheta = \begin{bmatrix}
			\! -2.0652  &  \! -4.7309  &  \! -0.6352 & \!0 & \!0 & \!0 & \!0 & \!0 & \!0 & \!0
		\end{bmatrix} \mathcal{F}(\zeta),
	\end{align*}
	or equivalently, the dynamic safety controller of the dt-GNS \eqref{eq: ex1_sys} is computed as
	\begin{align}
		u^+ \! = \!  \begin{bmatrix}
			\! -2.0652  &  \! -4.7309  &  \! -0.6352 & \!0 & \!0 & \!0 & \!0 & \!0 & \!0 & \!0
		\end{bmatrix} \! f(x, u). \label{des_u}
	\end{align}
	With the A-CBC, its corresponding level sets, and $c_a$ obtained, we can guarantee, according to Corollary \reff{Corollary 1}, that the A-dt-GNS \eqref{eq: ex1_sys_aug} remains safe over the time horizon of $T = 15$. This implies that the dt-GNS \eqref{eq: ex1_sys} is also safe over the specified time horizon, while its input constraint is enforced. This implementation was performed using Mosek in \textsc{Matlab} on a MacBook Pro (Apple M2 Max), completing in $0.24$ seconds.
	Figure \reff{fig:states} depicts the state and control input trajectories over the guaranteed time horizon. Moreover, to verify the correctness of our results, we assume access to the matrix $\mathcal{A}$ and check the satisfaction of condition \eqref{eq: A-con_bar_3}, as shown in Figure \reff{fig:con3}.
	
	\begin{figure}[t!]
		\centering
		\includegraphics[width=0.8\linewidth]{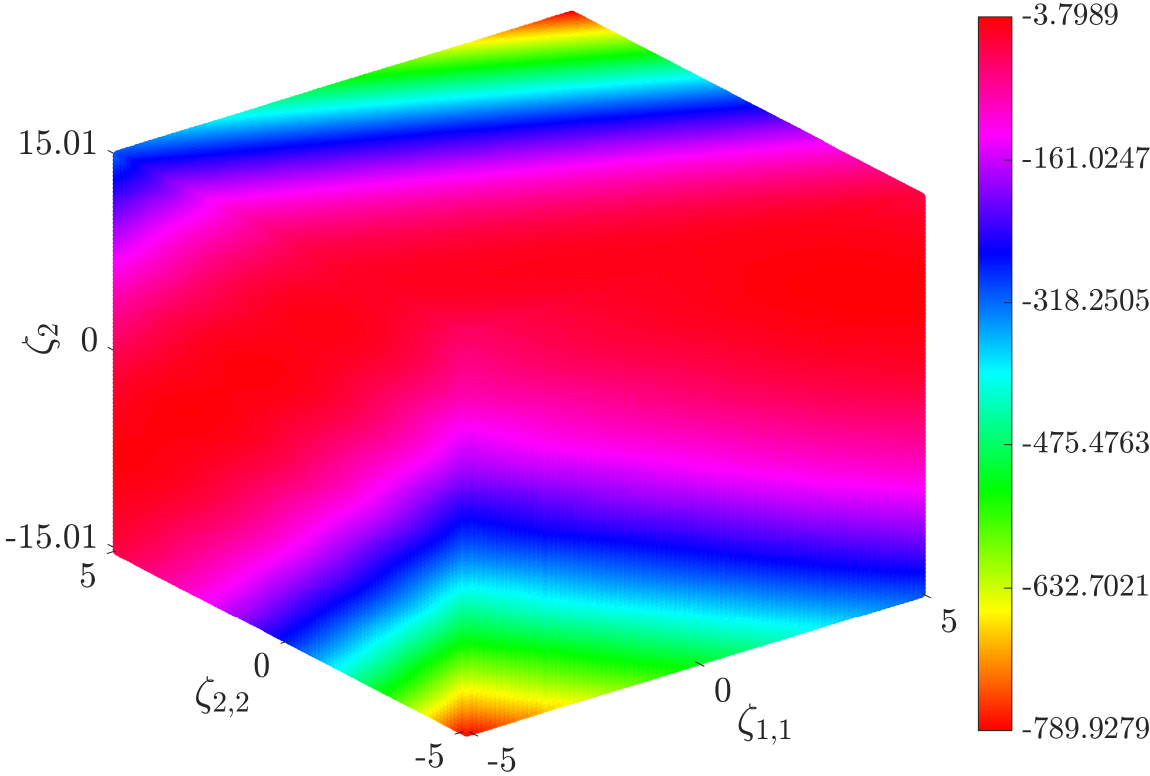}
		\caption{The heatmap shows the values of $\pmb{\mathds{B}_a}(\mathcal{A} \mathcal{F}(\zeta) + \mathcal{B}\vartheta) - \pmb{\mathds{B}_a}(\zeta) - c_a$, with $\mathcal{A}$ and $\mathcal{B}$ reported in \eqref{app3}, over the state set considered. As can be seen, these values are negative all over the state set, which means that condition \eqref{eq: A-con_bar_3} is fulfilled.}
		\label{fig:con3}
	\end{figure}
	
	To further boost the guarantee, we redesign the A-CBC and its associated safety controller using the alternative probabilistic method in Subsection \reff{Sub3}. We first set the violation parameter $\varepsilon = 0.01$ and confidence parameter $\beta = 10^{-10}$ (with $1 - \beta$ being very close to $1$). According to \eqref{eq: number}, the required assessment samples are computed as $\mathcal{N} = 6006$. After obtaining $\mathcal{Z}_2$ from \eqref{eq: con_thm_min} and \eqref{eq: con_thm_con1}, we solve the SCP that corresponds to the adjusted RCP \eqref{SCP_ad}, and, subsequently, obtain 
	\begin{align}
		&\mathcal{P} \!=\! \begin{bmatrix}
			0.0349 & 0.0203  &  0.0030\\
			0.0203  &  0.0575  &  0.0145\\
			0.0030  &  0.0145  &  0.0067
		\end{bmatrix}\!, c_a \!=\! 0.0078, \eta_a \!=\! 0.1334, \gamma_a \!=\!0.6483,\notag\\
		&\pmb{\mathds{B}_a}(\zeta) = 0.0349 \zeta_{1,1}^2 + 0.0575 \zeta_{1,2}^2 + 0.0067 \zeta_2^2 + 0.0406 \zeta_{1,1} \zeta_{1,2}\notag\\
		&\hspace{1cm} + 0.0060 \zeta_{1,1} \zeta_2 + 0.0290 \zeta_{1,2} \zeta_2,\notag\\
		& u^+ \! = \!  \begin{bmatrix}
			\! -0.4188  &  \! -2.1559  &  \! -0.0057 & \!0 & \!0 & \!0 & \!0 & \!0 & \!0 & \!0
		\end{bmatrix} \! f(x, u). \label{ff}
	\end{align}
	Thus, we can assert that the dynamic safety controller \eqref{ff} can make the system safe over the time horizon of $T = 66$, with a maximum risk of $1\%$ and a confidence of at least $1 - 10^{-10}$. This implementation took $2.38$ seconds to be completed. By utilizing this approach, we clearly extend the time horizon of the safety guarantee.

    \textbf{Second Case Study.} We now illustrate the scalability of our framework over a highly nonlinear system with $10$ state variables and $5$ control inputs, which evolves as follows:
    \begin{align}
		\Upsilon\!: \begin{cases}
			x^+_1 = x_1 + \tau\big(x_2 + \cos(u_1)\big),\\
			x^+_2  = x_2 + \tau\big(-x_1 + x_2 - \tanh{\big(1+x_1^2\big)} + u_1\big),\\
            x^+_3 = x_3 + \tau\big(x_4 + \sin(u_2)\big),\\
			x^+_4 = x_4 + \tau\big(-x_3 + x_4 - \sin\big(1+x_3^2\big) + u_2\big),\\
            x^+_5 = x_5 + \tau\big(x_6 + \cos(u_3)\big),\\
			x^+_6 = x_6 + \tau\big(-x_5 + x_6 - \tan^{-1}\big(1+x_5^2\big) + u_3\big),\\
            x^+_7 = x_7 + \tau\big(x_8 + \sin(u_4)\big),\\
			x^+_8 = x_8 + \tau\big(-x_7 + x_8 + \cos\big(1+x_7^2\big) + u_4\big),\\
            x^+_9 = x_9 + \tau\big(x_{10} + \cos(u_5)\big),\\
			x^+_{10} = x_{10} + \tau\big(-x_9 + x_{10} + \tan^{-1}\big(1+x_9^2\big) + u_5\big),
		\end{cases} \label{eq: ex2_sys}
	\end{align}
    with $\tau = 0.01$ being the sampling time. Each control input is constrained by $-30 \leq u_i \leq 30, \forall i \in \{1, 2, \dots, 5\}.$ For the sake of brevity, we assume that the available dictionary \eqref{eq: dictionary} includes only the terms that appear in the actual system's dynamics. 
    
    The regions of interest are given as $\mathcal{X} = [-10, 10]^{10}, \mathcal{X}_{\mathrm{o}} = [-2, 2]^{10},$ and $\mathcal{X}_1 = [-10, -7]^{10} \cup [7, 10]^{10}$. Similar to the previous case study and according to the procedure introduced in Subsection \reff{subsec:a}, we can obtain A-dt-GNS $\Upsilon_{\! a}$, which is omitted here due to the space limitation. Now, by setting $\epsilon_1 = \frac{1}{3000}$ and $\epsilon_2 = \frac{2999}{3000}$ (cf. Remark \reff{remark 2}), we have $\mathcal{X}^\zeta \coloneq [-10, 10]^{10} \times [-30.01, 30.01]^5,$ $\mathcal{X}^\zeta_{\mathrm{o}} \coloneq [-2, 2]^{10} \times [-0.1, 0.1]^5,$ and $\mathcal{X}^\zeta_{\mathrm{1}} \coloneq [-10, -7]^{10} \times [-30.01, 30.01]^5 \cup [7, 10]^{10} \times [-30.01, 30.01]^5 \cup [-10, 10]^{10} \times [-30.01, -30]^5 \cup [-10, 10]^{10} \\\times [30, 30.01]^5$.

  To obtain the A-CBC and its safety controller, we first gather input-state data as per \eqref{eq: data} over the time horizon $\mathcal{T} = 100.$ We also fix $\varpi=0.01.$ Utilizing our first proposed approach and after solving the SDP \eqref{eq: thm}, we obtain $c_a = 0.0566, \eta_a = 11.1984, \gamma_a = 42.9171,$ and the dynamic safety controller designed as
  
  \begin{align}\notag
  	u_1^+ &= -0.6692x_1   -2.3792x_2   -0.0268x_3   -0.0695x_4  +  0.1522x_5\\\notag
  	& \;\;\:\,   -0.3144x_6   -0.0208x_7   -0.1390x_8  +  0.0556x_9   -0.0774x_{10}\\\notag
  	&\;\;\:\,-1.0306u_1   -0.0444u_2 -0.0776u_3   -0.0292u_4   -0.0397u_5,\\\notag
  	u_2^+ &= -0.0396x_1   -0.1153x_2   -0.5529x_3   -2.6086x_4    +0.0525x_5\\\notag
  	&\;\;\:\,-0.0256x_6    +0.0532x_7   -0.0679x_8   -0.1685x_9   -0.1198x_{10}\\\notag
  	&\;\;\:\,-0.0344u_1   -1.1050u_2 -0.0219u_3   -0.0445u_4   -0.0093u_5,\\\notag
  	u_3^+ & = 0.2123x_1   -0.1845x_2    +0.0035x_3   -0.1587x_4   -0.7266 x_5\\\notag
  	&\;\;\:\, -2.3380 x_6  -0.1126 x_7  -0.0827x_8   -0.0481x_9   -0.1704 x_{10}\\\notag
  	& \;\;\:\, -0.0796u_1   -0.0317u_2   -1.0287u_3   -0.0145u_4   -0.0344u_5,\\\notag
  	u_4^+ &= 0.0107x_1   -0.0948x_2   -0.0110  x_3 -0.2080x_4   -0.0628 x_5\\\notag
  	&\;\;\:\, -0.0234x_6   -0.7123 x_7  -2.4416 x_8  -0.2346 x_9  -0.1244 x_{10}\\\notag
  	&\;\;\:\,-0.0383u_1   -0.0586u_2   -0.0195 u_3  -1.0373  u_4 -0.0128u_5,\\\notag
  	u_5^+ &= 0.0438x_1   -0.1480 x_2  -0.1271x_3    +0.0126 x_4   +0.0417   x_5\\\notag
  	&\;\;\:\,-0.1093x_6   -0.2061x_7    +0.0015 x_8  -0.8210 x_9  -2.3757 x_{10}  \\\notag
  	&\;\;\:\,-0.0497 u_1  -0.0137 u_2  -0.0445 u_3  -0.0140 u_4  -1.0191u_5.
  \end{align}
   Under this designed controller, we guarantee that the system \eqref{eq: ex2_sys} remains safe over the time horizon of $T = 560.$ This case study clearly demonstrates that our framework is scalable when coping with highly nonlinear systems with a high number of state variables. The implementation was completed in $0.84$ seconds, offering applicability to runtime implementation.
    	
\section{Conclusion and Discussion}\label{Section: Conc}
Our paper developed a direct data-driven methodology for discrete-time general nonlinear systems, enabling the joint learning of CBCs and dynamic controllers to uphold safety properties under input constraints. Although our approach provides safety guarantees within a finite time horizon, it is capable of handling \emph{general nonlinear systems} with unknown dynamics—a highly challenging task. Even when the model is known, model-based computational tools like SOS struggle with our case studies due to the presence of nonlinear terms such as $\sin(u)$ and $\ln\big(1 + x_1^2\big)$. In such cases, model-based SMT solvers are required; however, despite requiring precise knowledge of models, they may encounter termination issues~\cite{wongpiromsarn2015automata} or scalability challenges depending on system complexity~\cite{abate2022neural}. On a further note, the finite horizon guarantee in our work is not restrictive, as we solved the first complex problem with a guarantee of $T=15$ in only $0.24$ seconds. This implies that, if a longer guarantee is required, our scalable LMI conditions can be re-solved every $15$ time steps to obtain an updated guarantee \emph{in under a second}, allowing the conditions to be solved even runtime. As a future direction, our proposed framework can be extended to scenarios with exogenous, bounded disturbances, aiming to design safety controllers that ensure robustness against \emph{unknown-but-bounded} disturbances. Another direction for future work is to explore a solution that guarantees \emph{infinite-horizon} safety.
	
\bibliographystyle{ieeetran}
\bibliography{biblio}
	
\end{document}